\documentclass[11pt]{article}
\usepackage[utf8]{inputenc}
\usepackage[margin=1in]{geometry}
\usepackage{amsmath,amsthm,amssymb}
\usepackage{xcolor}
\usepackage{hyperref}
\usepackage[nameinlink,capitalise]{cleveref}
\usepackage{natbib}
\usepackage{float}
\usepackage{bbm}
\usepackage{mathtools,dsfont}
\usepackage{pgfplots}
\usepackage{doi}
\usepackage{tikz}
\usepackage{thm-restate}
\usepackage[ruled, linesnumbered]{algorithm2e} 

\newcommand{\E}{\mathbb{E}}
\newcommand{\mathsc}[1]{{\normalfont\textsc{#1}}}

\newcommand{\nd}{\mathcal{D}}
\newcommand{\ns}{\mathcal{S}} 
\newcommand{\nss}[1]{\mathcal{S}^{\mathbbm{1}}_{#1}} 

\newcommand{\nm}{m}
\newcommand{\nz}{\mathcal{Z}}

\newcommand{\nzz}{\mathcal{Z}_0} 
\newcommand{\nzm}{\mathcal{Z}_0'} 
\newcommand{\nzf}{\mathcal{Z}_1} 
\newcommand{\nzeo}{
\mathcal{Z}_1^{\mathsf{Ex}}
} 
\newcommand{\nzet}{
\mathcal{Z}_2^{\mathsf{Ex}}
} 

\newcommand{\mg}{m^{\mathrm{G}}} 
\newcommand{\mt}{m^{\mathrm{T}}}  
\newcommand{\vg}{v^{\mathrm{G}}} 
\newcommand{\vt}{v^{\mathrm{T}}}  

\newcommand{\tg}{t^{\mathrm{G}}} 
\renewcommand{\tt}{t^{\mathrm{T}}}  

\newcommand{\cs}{\mathcal{CS}} 
\newcommand{\nr}{\mathcal{R}} 
\newcommand{\W}{\mathcal{W}}
\newcommand{\pf}{\mathsc{SPrefix}} 
\newcommand{\pfv}{\mathsc{IPrefix}} 

\newcommand{\mks}{F_{\nd}(v_k)} 

\newcommand{\ts}{\tilde{s}}

\renewcommand{\th}{\text{th}}
\renewcommand{\d}{\,\mathrm{d}}
\Crefname{algocf}{Algorithm}{Algorithms}
\newtheorem{theorem}{Theorem}[section]
\newtheorem{lemma}[theorem]{Lemma}

\newtheorem{proposition}[theorem]{Proposition}
\theoremstyle{definition}
\newtheorem{definition}[theorem]{Definition}
\newtheorem{example}[theorem]{Example}

\newcommand{\eps}{\varepsilon}
\DeclareMathOperator*{\argmax}{arg\,max}

\newcommand{\algname}{Split-and-Match\,}

\definecolor{myblue}{rgb}{0.15, 0.1, 0.95}
\definecolor{mygreen}{rgb}{0.15, 0.65, 0.25}
\definecolor{myred}{rgb}{0.75, 0.25, 0.15}

\hypersetup{
    colorlinks=true,
    linkcolor=myred,
    citecolor = myblue,      
}

\title{Fair Price Discrimination}

\makeatletter
\def\@fnsymbol#1{\ensuremath{\ifcase#1\or \dagger\or \ddagger\or
   \mathsection\or \mathparagraph\or \|\or **\or \dagger\dagger
   \or \ddagger\ddagger \else\@ctrerr\fi}}
\makeatother
\newcommand*\samethanks[1][\value{footnote}]{\footnotemark[#1]}

\author{Siddhartha Banerjee\thanks{School of Operations Research and Information Engineering, Cornell University, Ithaca, NY 14850. Email: \texttt{sbanerjee@cornell.edu}.} \and Kamesh Munagala\thanks{Department of Computer Science, Duke University, Durham, NC 27708-0129. Emails: \texttt{kamesh@cs.duke.edu}, \texttt{yiheng.shen@duke.edu}. Supported by NSF grant CCF-2113798.} \and Yiheng Shen\samethanks[2] \and Kangning Wang\thanks{Department of Computer Science, Stanford University, Stanford, CA 94305. Email: \texttt{knwang@stanford.edu}.}}
\date{}

\begin{document}

\maketitle

\begin{abstract}
A seller is pricing identical copies of a good to a stream of unit-demand buyers. Each buyer has a value on the good as his private information. The seller only knows the empirical value distribution of the buyer population and chooses the revenue-optimal price. We consider a widely studied third-degree price discrimination model where an information intermediary with perfect knowledge of the arriving buyer's value sends a signal to the seller, hence changing the seller's posterior and inducing the seller to set a personalized posted price. Prior work of Bergemann, Brooks, and Morris (American Economic Review, 2015) has shown the existence of a signaling scheme that preserves seller revenue, while always selling the item, hence maximizing consumer surplus. In a departure from prior work, we ask whether the consumer surplus generated is fairly distributed among buyers with different values. To this end, we aim to maximize welfare functions that reward more balanced surplus allocations.

Our main result is the surprising existence of a novel signaling scheme that \emph{simultaneously} $8$-approximates \emph{all} welfare functions that are non-negative, monotonically increasing, symmetric, and concave, compared with any other signaling scheme. Classical examples of such welfare functions include the utilitarian social welfare, the Nash welfare, and the max-min welfare. Such a guarantee cannot be given by any consumer-surplus-maximizing scheme -- which are the ones typically studied in the literature. In addition, our scheme is socially efficient, and has the fairness property that buyers with higher values enjoy higher expected surplus, which is not always the case for existing schemes.
\end{abstract}


    

\section{Introduction}
Imagine a seller with infinite supply of a good. They wish to sell to a population of unit-demand buyers with standard quasi-linear utilities. The seller knows the empirical distribution $\nd$ of the buyer valuations and chooses a revenue-maximizing price to sell the good. In this paper, we consider this simple setting, but with a twist: there is an additional information intermediary who can segment the market and help the seller price-discriminate. Information intermediaries for price discrimination were first considered by \citet*{Bergemann15} and our work proposes and studies new desiderata for them. 

Such intermediaries are becoming ubiquitous in modern two-sided e-commerce platforms. Consider for example ad exchanges~\citep{doubleclick,verizon,msads,pubmatic}, where the platform acts as an intermediary between buyers (in this case advertisers) and sellers (in this case, publishers controlling the ad slot). The intermediary wants the best for both sides; however, as in classical auctions, the seller -- not the intermediary -- controls the price at which trade happens. Other examples include retail platforms such as Amazon marketplace, who also effectively serve as intermediaries -- they merely facilitate the trade, but do not control the prices.

In such settings, the platform can use machine learning and its vast trove of data on buyer behavior to accurately predict the value of buyers. It can then choose to reveal information about the current buyer to the seller in order to influence the trade. This information (or \emph{signal}) leads to the seller updating its prior $\nd$ over buyer values to a posterior distribution over values given the signal. The seller now posts the optimal (revenue-maximizing) price for this posterior. Such information revelation is termed signaling or {\em third-degree price discrimination}\footnote{It is termed ``third-degree'' price discrimination because the seller or intermediary divides the market into segments, each with its own price. In contrast, in first-degree price discrimination, the seller has perfect information and charges buyers exactly their value, while in second-degree price discrimination, the seller sells similar yet `different' goods (differing in quality/quantity) at different prices.} and is a special case of {\em Bayesian persuasion}~\citep{Kamenica11}. Note that in practice, the seller and the intermediary could be the same entity, such as a retail or ride-share platform that wants to use buyer information to segment the market and perform price discrimination. 

To understand this setting better, consider two extremes: At one extreme, the intermediary can choose to reveal no information to the seller, in which case the seller's posterior remains $\nd$. Therefore, the seller posts the Myerson price~\citep{Myerson81} $p^{\mathsf{My}} = \argmax_p p \cdot \Pr_{v \sim \nd}[v \ge p]$ on $\nd$ and raises revenue $R^{\mathsf{My}} = p^{\mathsf{My}} \cdot \Pr_{v \sim \nd}[v \ge p^{\mathsf{My}}]$. Since trade does not happen if the buyer's private value $v$ is below $p^{\mathsf{My}}$, this scheme is generally {\em inefficient} -- the consumer (buyer) surplus $\E_{v \sim \nd}[(v-p^{\mathsf{My}})^+]$ plus seller revenue $R^{\mathsf{My}}$ is less than the maximum possible total surplus, $\E_{v \sim \nd}[v]$. 

At another extreme is full information revelation or first-degree price discrimination, where the intermediary reveals the actual buyer value $v$ to the seller. In this case, the seller's posterior collapses to the deterministic value $v$. The seller can now post price slightly below $v$, so that trade always happens. However, this efficiency comes at a cost -- the buyer now obtains zero surplus (their value minus price paid), while the seller's revenue becomes equal to the total surplus, $\E_{v \sim \nd} [v]$. Note that in the no-signaling case discussed above, the consumer surplus $\E_{v \sim \nd}[(v-p^{\mathsf{My}})^+]$ could be positive -- thus between these two schemes, from a utilitarian point of view, no-signaling is better for the buyers, while full-revelation is better overall.

\subsection{Optimal Signaling and Fairness}
Signaling clearly helps the seller since they can always obtain at least as much revenue $R^{\mathsf{My}}$ as in no-signaling (e.g. by ignoring the signal). What is less clear is whether signaling can improve consumer surplus at all. In a remarkable result, \citet*{Bergemann15} showed the existence of ``buyer-optimal'' signaling schemes in the following sense: The seller's expected revenue remains the same as in no-signaling (i.e., $R^{\mathsf{My}}$, which is the \emph{minimum} possible under any signaling scheme), while trade is always efficient (i.e., the item always sells), which means that the sum of the consumer surplus and the seller revenue is the maximum possible total surplus, $\E_{v \sim \nd} [v]$. Hence, the consumer surplus must be as large as it could possibly be.

This is a beautiful result, but is unsatisfying upon closer inspection. Note that while the proposed scheme maximizes consumer surplus, it is not the unique such scheme~\citep{Bergemann15,Cummings20,KOM22}. Are all ways of splitting this aggregate surplus among buyers equal, even if this gain in surplus is ``subsidized'' more by a particular group of buyers? We think not -- maximizing the utilitarian total consumer surplus should not be the sole consideration; it is natural to also desire that \emph{price discrimination is fair at the level of individual buyers} -- 
but how should we formalize this?

A first idea is to require some form of \textbf{monotonicity} in the surplus split. Let $\cs_v$ be the expected consumer surplus that a signaling scheme provides to a buyer of value $v$; we could now require that $\cs_v \le \cs_{v'}$ whenever $v < v'$. This is true in the absence of signaling (as fixed pricing is monotone), and so should perhaps be expected to hold in the presence of signaling. It also captures some sense of envy in price discrimination -- a buyer with larger value should not envy the surplus seen by a buyer with smaller value. We show via examples in~\cref{ssec:metrics} that even this very natural constraint rules out some existing buyer-optimal schemes.\footnote{We note that the scheme for continuous priors in~\citep{Bergemann15} is both monotone and buyer-optimal.} 

An alternative and more wide-reaching fairness requirement is given by the following paradigm:

\paragraph{Universal Welfare Maximization (and Majorization). } Consider the surplus vector where its $j^\text{th}$ dimension is the expected surplus of the $j^\text{th}$ buyer. 
A general welfare function takes the surplus vector as input, and outputs a non-negative real number (higher is better). We restrict to welfare functions that are symmetric, non-decreasing, and concave: Symmetry ensures equal treatment to all buyers; non-decreasing ensures that Pareto improvements are desirable; and concavity is a common fairness consideration favoring balanced allocations. Common examples of such welfare functions include the utilitarian social welfare function, the Nash welfare function, and the max-min welfare function. 
A fair signaling scheme could be defined as one which maximizes such a welfare function; however, it is unclear how to unambiguously pick one welfare function among the numerous possibilities. 

What would be ideal is if there is a \emph{universal} scheme that is optimal (or at least, approximately optimal) for \emph{all} such welfare functions. This universal maximization of concave functions is closely related to \emph{majorization}~\citep*{karamata1932inegalite,hardy1952inequalities} and its approximate form: $\alpha$-majorization (see e.g.~\citep*{GoelMP05,GoelM06}).

\subsection{Our Results}
The main question we ask is: 
\begin{quote}
    In third-degree price discrimination, how close can we get to a \emph{universally-fair} signaling scheme, i.e., one which is (monotone and) near-optimal for any welfare function?
\end{quote}
At the outset, one might be pessimistic: For resource allocation and stochastic optimization problems~\citep*{GoelMP05,KumarK06,ChakrabartyS19}, typically $\alpha = \omega(1)$, where $\alpha$ is the approximation factor for majorization (and hence for universal welfare maximization). Indeed, as we show in \cref{sec:lb}, any buyer-optimal signaling scheme in the sense of~\citet*{Bergemann15} cannot be $\alpha$-majorized for {\em any} given constant $\alpha$. Given this, one may wonder if universal welfare maximization is too strong a condition to expect.

Our main result is a surprising new signaling scheme that shows the following theorem:

\begin{theorem}[formally stated as \cref{thm:main}]
\label{thm:main0}
For any prior  $\nd$, there is a signaling scheme that is $8$-majorized, and hence it simultaneously $8$-approximates all non-negative, increasing, symmetric and concave welfare functions, compared with any other signaling scheme. Further, this scheme is monotone and efficient, and can be computed in time polynomial in the size of the support of $\nd$.
\end{theorem}
Our main theorem therefore shows that we can be (near)-universally-fair (i.e., near-optimal for any welfare function). This signaling scheme sacrifices some consumer surplus to achieve this guarantee; however, as mentioned above, this sacrifice is necessary -- as we show in \cref{sec:lb}, any exactly buyer-optimal scheme is not $\alpha$-majorized for any constant $\alpha$. Further note that by definition, our scheme is also $8$-approximately buyer-optimal. We complement our $8$-approximation with a lower bound of $1.5$ in \cref{sec:lb}: There is no signaling scheme that is $\alpha$-majorized by every other signaling scheme for $\alpha < 1.5$.

At a technical level, the proof of \cref{thm:main0} constructs a very different signaling scheme from prior work on price discrimination. The scheme is composed of signals such that each of them induces a posterior as a distribution over at most two values. We first decompose the prior $\nd$ into a collection of such signals, and show that this collection $4$-approximates prefix sums of consumer surplus when sorted on buyer value. We then apply a novel ironing procedure to modify the signals so that the resulting scheme is approximately majorized, while losing an additional factor of $2$. Both steps are non-trivial, and together yields an $8$-majorized scheme that is also monotone (and socially efficient).

\subsection{Related Work}
Our model of third-degree price discrimination is a special case of {\em information design} (see~\citep{bergemann2019information}) where an information mediator provides information to impact the behavior of agents. This has also been termed {\em signaling} or {\em persuasion} in literature. (See \citep{dughmi2017algorithmic}.) In {\em Bayesian persuasion} first proposed by \citet{Kamenica11}, there is one agent called the {\em receiver} who receives additional information from a better-informed {\em sender}. Given the signal, the receiver computes their posterior over the state of nature and chooses an action to maximize their own utility. The sender can design the signals so that the receiver, acting in her own interest, maximizes some utility function the sender cares about. This problem has been widely studied in various contexts~\citep{dughmi2016persuasion, dughmi2019algorithmic,Babichenko21,Bergemann15,chakraborty2014persuasive,xu2015exploring,ParetoIS}.

In the setting we consider, the sender is the intermediary, while the receiver is the seller that maximizes their revenue given the signal. This was first considered by~\citet*{Bergemann15}, who showed buyer-optimal signaling schemes that preserve seller revenue while transferring the rest of the surplus to the buyers. Subsequently, it was shown by~\citet*{Cummings20} that the set of all buyer-optimal signaling schemes can be specified by a linear program. Several works~\citep{shen2018closed,cai2020third,DBLP:conf/innovations/MaoLW22,BergemannDLZ22,DBLP:conf/sigecom/AlijaniBMW22,KOM22} consider various extensions to the basic single seller/buyer setting, and show exact/approximate buyer-optimality under various assumptions. 

The concept of majorized vectors has existed for a long time~\citep*{karamata1932inegalite,hardy1952inequalities}, and is equivalent to solutions that simultaneously maximize symmetric concave functions of the coordinates. In the context of resource allocation and routing problems, an approximate version of this concept was defined by \citet*{GoelMP05}, and subsequently shown by \citet*{GoelM06} to be equivalent to solutions that simultaneously approximately maximize every symmetric concave function of the coordinates; see also~\citep*{KumarK06,ChakrabartyS19}. It was shown by \citet*{GoelMP05} that the best approximation factor is the solution of a linear program. However, the approximation factor is problem-dependent and typically logarithmic in the number of coordinates for general routing problems. The surprising aspect of our paper is that this factor is only a constant for the price discrimination problem, and is achieved by a very simple signaling scheme. This is similar in spirit to recent results in metric distortion of voting rules~\citep*{GoelKM17}, where it is shown that the Copeland rule is $5$-majorized by any other rule.

\citet{XuRegulation} consider fairness in price discrimination by imposing a bound on the ratio or difference in prices that the seller is allowed to charge to different buyers (akin to monotonicity). They assume a perfectly informed seller (first-degree price discrimimation) and derive a unique optimal pricing strategy as well as characterize the tradeoffs for different buyer value distributions. Our work in contrast focuses on the more involved objective of majorization, and furthermore, we do not assume the seller is perfectly informed (third-degree price discrimination).

Our work connects to the larger body of work on fairness in machine learning, where again, optimality in the sense of overall risk minimization (ERM) can lead to systematic unfairness~\citep{Kearns1,Kearns2,multicalib,DworkIndividual,Anilesh}. Much of this work focuses on the tradeoffs between efficiency and fairness. As machine learning systems become more pervasive, it becomes important to consider not just their direct impact, but also their impact to downstream applications when they are embedded in a larger system. In our case, the larger system is a marketplace platform that uses machine learning to predict buyer values and help sellers price-discriminate. Our results show that na\"ively maximizing surplus can be unfair, while different mechanisms can achieve good tradeoffs between efficiency and fairness.

\section{Preliminaries}

\subsection{Basic Setting}

\paragraph{Seller and Buyers.}
A monopolistic seller of a good has infinite supply, and wants to price them so as to maximize her revenue. 
There are a finite number of buyers in the market. 
Each buyer is interested in buying at most one copy of the good, and has a value for the good given by a positive real number. A buyer chooses to buy if and only if the price is at most his value.
We henceforth focus on discrete empirical distributions over buyer valuations; in particular, we consider distributions with support size $n$ over values $v_1<v_2<\cdots<v_n$. (We write $v_0 := 0$ to simplify notations.) For any distribution $\mathcal{P}$, we use $f_{\mathcal{P}}(v)$ to denote the probability mass function: $f_{\mathcal{P}}(v) := \Pr_{v'\sim {\mathcal{P}}}[v' = v]$, and define the cumulative distribution function (CDF) $F_{\mathcal{P}}(v) := \Pr_{v'\sim {\mathcal{P}}}[v'\le v]$ and complementary CDF $G_{\mathcal{P}}(v) := \Pr_{v'\sim {\mathcal{P}}}[v'\ge v]$.

Let $\nd$ denote the empirical distribution of buyer valuations. The seller knows the distribution $\nd$, but not the actual value of each buyer. Consequently, without additional information, the seller chooses a common price $p^{\mathsf{My}}$ (sometimes called the Myerson price~\citep*{Myerson81}) for all buyers such that $p = p^{\mathsf{My}}$ maximizes $p \cdot G_{\nd}(p)$.

\paragraph{Price Discrimination via an Information Intermediary.}
The main idea in the work of \citet*{Bergemann15} is that in this setting, one can model the effects of \emph{price discrimination} by considering an exogenous intermediary who provides some additional signal to the seller about each buyer, enabling the seller to modify the price offered to that buyer. We now formalize this as a game among the intermediary, the seller and the buyers.

We assume the information intermediary knows $\nd$ as well as the exact value of each buyer. Independently, for each buyer, the intermediary sends a \emph{signal} about the buyer's value to the seller via some chosen \emph{signaling scheme}: a (potentially randomized) mapping from a value in $\{v_1,v_2,\ldots,v_n\}$ to some set of signals $[Q]$. Crucially, the intermediary commits to a scheme upfront, and the scheme is known to the seller.

From the perspective of the seller, since all agents are a priori indistinguishable, the effect of receiving a signal is to update the seller's belief over the buyer's value from $\nd$ to some new distribution $\ns_q$ ($q\in[Q]$) over possible values. Consequently, with a slight abuse of notation, instead of defining a signaling scheme in terms of the mapping from value to signal, we directly define it in terms of the resulting posterior distributions $\ns_q$ corresponding to each $q$, as well as the resulting distribution over these signals. Formally:

\begin{definition}[\textbf{Signal; Signaling Scheme}]\label{def:sig}
A \emph{signal} $\ns$ updates the seller's belief over a buyer's value from $\nd$ to some new distribution $\ns$. A \emph{signaling scheme} $\nz=\{(\ns_q, \gamma_q)\}_{q \in [Q]}$ is a collection of $Q$ signals $\{\ns_q\}_{q \in [Q]}$ with weights $\{\gamma_q\}_{q \in [Q]}$ that satisfy:
(1) $\gamma_q \geq 0$ and $\sum_{q \in [Q]} \gamma_q = 1$; and moreover
(2) $\sum_{q \in [Q]}\gamma_q \ns_q = \nd$.
\end{definition}
We note again that $\ns_q$ denotes both the $q^{\th}$ signal in the signaling scheme, and the posterior of the seller after receiving the $q^{\th}$ signal. The constraints in \cref{def:sig} ensure that the signaling scheme $\nz$ is \emph{Bayes plausible}~\citep*{Kamenica11}, i.e., that the expected posterior is equal to the prior. Given a signaling scheme as defined above, it is easy to construct the random mapping from values to signals: each $v$ is mapped to $\ns_q$ with probability $\gamma_qf_{\ns_q}(v)/f_{\nd}(v)$.

\paragraph{Outcomes under Signaling.}
After receiving signal $\ns_q$ from the intermediary, the seller offers the buyer a new price $p_{\ns_q}^*$ based on the new posterior $\ns_q$ satisfying
\[
p_{\ns_q}^*=\argmax_{v} v \cdot G_{\ns_q}(v).
\] 

The resulting expected gains from trade are split between the buyer and the seller as:
\begin{itemize}
	\item \emph{Producer surplus} (or revenue) of the seller: $\nr(\ns_q) = \max_{v} v \cdot G_{\ns_q}(v) = p_{\ns_q}^*\cdot G_{\ns_q}(p_{\ns_q}^*).$
	\item \emph{Consumer surplus} of the buyer with value $v$: $
	\cs_v(\ns_q)=\mathds{1}[v \ge p_{\ns_q}^*]\cdot (v-p_{\ns_q}^*).$	
\end{itemize} 

In the event that $v \cdot G_{\ns_q}(v)$ are maximized at multiple points, we assume the seller breaks ties by choosing the lowest tied price.\footnote{Note that we can avoid ties by slightly perturbing each signal, without changing the message of our results.} Moreover, we can now define the \emph{expected} outcomes under a given signaling scheme: The expected consumer surplus of a buyer with value $v$ is the expectation of that buyer's surplus on all signals that the seller might receive from the intermediary. Similarly we can define the expected seller revenue.
\begin{definition}[Expected Outcomes under Signaling]\label{def:average_surplus}
Given a signaling scheme $\nz=\{(\ns_q, \gamma_q)\}_{q \in [Q]}$, the \emph{expected consumer surplus} of a buyer with value $v$ under $\nz$ is:
\begin{equation}
\cs_v(\nz)=\sum_{q\in[Q]} \cs_v(\ns_q)\cdot\frac{\gamma_q\cdot f_{\ns_q}(v)}{f_{\nd}(v)}.\label{eqn:scheme_surplus}
\end{equation}
Moreover, the overall expected consumer surplus under $\nz$ is
$\cs(\nz) = \sum_{q\in[Q]} \gamma_q \cdot f_{\ns_q}(v) \cdot \cs_v(\ns_q)$.

Similarly, the {\em seller's expected revenue} is given by: 
\begin{equation}
\nr(\nz)=\sum_{q} \nr(\ns_q)\cdot\gamma_q.\label{eqn:scheme_revenue}
\end{equation}
\end{definition}
We illustrate our setting, signaling schemes, and the above metrics with the following running example; ~\cref{fig:non_monotone,fig:monotone_major} show different signaling schemes for this setting.

\begin{example}[Running example]
\label{ex:non_monotone}
The buyer values are given by $\langle 1,2,5,6\rangle$ with distribution on this support being $\nd = \langle 0.25, 0.25, 0.25, 0.25 \rangle$. The revenue under each of the posted prices is $\langle 1, 1.5, 2.5, 1.5 \rangle$, and thus $5$ is the Myerson price under $\nd$, resulting in revenue $\nr^{\mathsf{My}}=2.5$.

\cref{fig:non_monotone} illustrates one particular signaling scheme $\nzeo$ for this setting (based on the construction of~\citet*{Bergemann15}). Here, it is easy to check that $5$ is an optimal price in all the signals, resulting in (seller-optimal) revenue of $\nr(\nzeo)= 5\cdot \left(f_{\nd}(5)+f_{\nd}(6)\right) = 2.5$. To compute the consumer surplus, take $v_3=5$ as an example: the expected consumer surplus of a buyer with value $5$ is \[\cs_{5}(\nzeo) = \left(\frac{1}{60} \cdot (5 - 1) + \frac{1}{90} \cdot (5 - 2)\right) \Big/ \left
(\frac{1}{4}\right) = 0.4.\] 
\end{example}

\subsection{Global and Per-Agent Performance Metrics of Signaling Schemes}
\label{ssec:metrics}

\begin{figure}[!t]
    \centering  \includegraphics{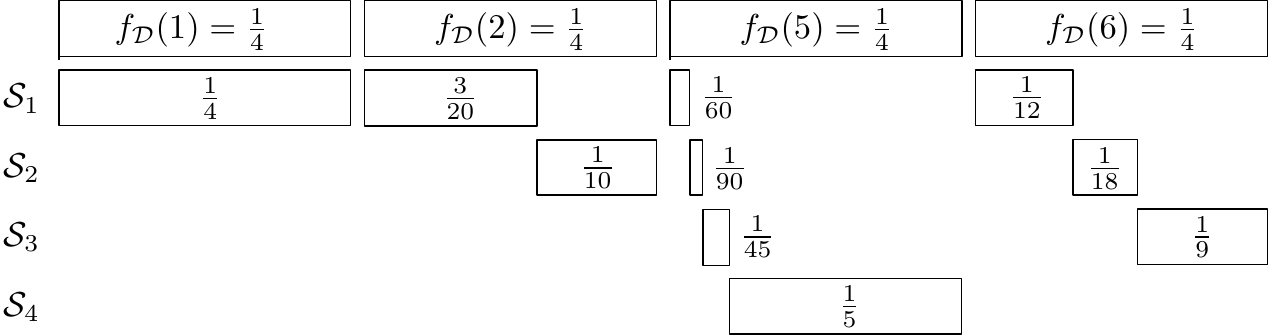}
    \caption{\small\it Signaling scheme in \cref{ex:non_monotone}: The distribution $\nzeo$ is drawn as rectangles in the first row. Each subsequent row corresponds to a signal under $\nz$. The rectangles under each $v_i$'s column indicate the mass $\gamma_q f_{\ns_q}(v_i)$ placed on $v_i$ in each signal. 
    We can recover each signal and its weight by normalizing; for example, signal $\ns_1$ satisfies $f_{\ns_1}(1) = \frac{1}{2}$, $f_{\ns_1}(2) = \frac{3}{10}$, $f_{\ns_1}(5) = \frac{1}{30}$, $f_{\ns_1}(6) = \frac{1}{6}$ and has weight $\gamma_1 = \frac{1}{2}$.
}
\label{fig:non_monotone}
\end{figure}
\begin{figure}[htbp]
    \centering  \includegraphics{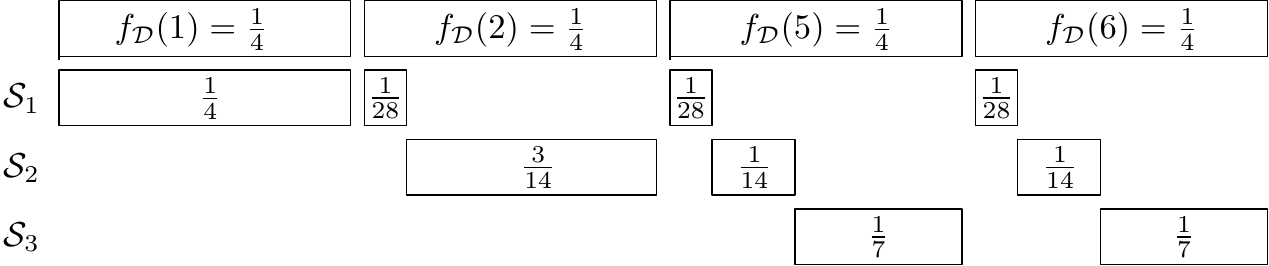}
    \caption{\small\it A completely different signaling scheme $\nzet$ for the instance given in \cref{ex:non_monotone}: The notations are the same as those in \cref{fig:non_monotone}. 
    Note that while the two schemes are very different, they both are efficient (i.e., the item is always sold), and have the same revenue and same overall consumer surplus (although the distribution of the overall surplus between different values is different; see \cref{ex:monotone_major}).
}
\label{fig:monotone_major}
\end{figure}

\paragraph{Utilitarian Metrics (Efficiency and Consumer Surplus).} For a signaling scheme $\nz$ to be efficient (i.e., to maximize the gains from trade), it needs to ensure the item is always sold. 
This corresponds to requiring that for each signal $q\in[Q]$, the optimal price posted by the seller under $\ns_q$ is the smallest value with non-zero probability in the support of $\ns_q$. If this holds, then any buyer will always accept the posted price and the item is always sold. 

Given any discrete distribution $\mathcal{P}$, the \emph{lowest value in the support} of $\mathcal{P}$ is defined as $\underline{v}_{\mathcal{P}} := \min\{v \mid f_{\mathcal{P}}(v)>0\}$. 
Note that since we focus on distributions over a finite support, the minimum exists. Now we can formally define the condition for a signaling scheme to be efficient:

\begin{definition}[Efficient Signaling Scheme]
A signaling scheme $\nz=\{(\ns_q,\gamma_q)\}_{q\in[Q]}$ is \emph{efficient} if
\[
\underline{v}_{\ns_q} = \argmax_{v} v\cdot G_{\ns_q}(v) \qquad \forall\, q\in [Q].
\]
\end{definition}

While efficiency ensures that a signaling scheme maximizes the overall gains-from-trade, it does not specify how the surplus is divided. In particular, revealing the buyer's true value to the seller is an efficient scheme, but results in the seller getting the full surplus. An alternative is to maximize the overall consumer surplus $\cs(\nz)$. From the above definitions, it is easy to see that given $\nd$ one can write the problem of constructing a signaling scheme $\nz$ that maximizes $\cs(\nz)$ via a linear program; it is not clear however what guarantees such a scheme has, or even, if it is efficient. The surprising result of \citet*{Bergemann15} is that this is indeed the case:

\begin{proposition}[From~\citep*{Bergemann15}]
\label{prop:BBM}
For any given $\nd$, let $\nr^{\mathsf{My}}=\max_v v\cdot G_{\nd}(v)$ denote the optimal revenue without signaling (i.e., the ``Myerson revenue''). Then there exist efficient signaling schemes $\nz$ under which $\nr(\nz)=\nr^{\mathsf{My}}$.
\end{proposition}

Note that since the seller can always get $\nr^{\mathsf{My}}$ under any signaling scheme (by ignoring the signal and posting $p^{\mathsf{My}}$), and since the signaling scheme is efficient, it must have maximized $\cs(\nz)$. We also note that \citet*{Bergemann15} in fact construct an explicit signaling scheme that achieves this result; since then, alternative constructions have been found~\citep{Cummings20,KOM22} which also maximize $\cs(\nz)$, with other additional desirable properties; moreover, any convex mixture of such schemes leads to new signaling schemes which all maximize $\cs(\nz)$.

\begin{example} \label{ex:monotone_major}
Continuing our running example from \cref{ex:non_monotone}, note that under the signaling scheme $\nzeo$, the item is always sold, meaning that $\nzeo$ is efficient. Moreover, since $\nr(\nzeo)=\nr^{\mathsf{My}}$, we have that $\nzeo$ maximizes $\cs(\nz)$, and the maximum surplus consumer is $1$.

\cref{fig:monotone_major} illustrates an alternative efficient signaling scheme $\nzet$ (based on the construction of~\citet*{KOM22}), which we call ``remove from bottom''. Again, we can check that $\cs(\nzet)=1$ and the item always sells, so that this scheme is buyer-optimal.
\end{example}

\paragraph{Fairness Metrics (Monotonicity and Equitable Welfare Functions).}
The main problem with focusing on utilitarian metrics alone is that they do not give good guarantees for each individual agent's surplus. To understand how \emph{fair} a given signaling scheme is, we need to consider additional performance metrics. The simplest of these is monotonicity: we say a signaling scheme $\nz$ is \emph{monotone} if buyers with larger values gain larger expected surplus from $\nz$: 
\begin{definition}[Monotonicity]
	A signaling scheme $\nz$ is \emph{monotone} if for any ordered pair of values $v_i < v_j$, we have $\cs_{v_{i}}(\nz) \le \cs_{v_{j}}(\nz)$.  
\end{definition}

Our running example shows that not all schemes satisfying \cref{prop:BBM} (i.e., efficient and consumer surplus maximizing) are monotone. 
\begin{example}
In \cref{ex:monotone_major} under $\nzet$, the expected consumer surplus of buyers with values $\langle 1,2, 5,6\rangle$ are respectively $\langle 0, \frac{1}{7}, \frac{10}{7}, \frac{17}{7}\rangle$.  This vector is monotone; 
however, $\nzeo$ in \cref{ex:non_monotone} has surplus vector $\langle 0, 0.6, 0.4, 3 \rangle$, which is not monotone since $\cs_2(\nzeo) = 0.6 > 0.4 = \cs_5(\nzeo)$. 
\end{example}

What can we say about what a \emph{fair} signaling scheme is, beyond the above metrics (efficiency, maximizing consumer surplus, monotonicity)? One option that is often used is to maximize an alternative \emph{equitable} welfare function -- one which promotes a more balanced solution. Such a welfare function $\W$ takes as input the surplus vector $\vec{u}$ containing the expected surplus under each value, and outputs a real number; moreover, $\W$ satisfies the following natural properties:
\begin{itemize}
    \item (Symmetry) For any $\vec{u}$ and any permutation $\sigma$, $\W(\sigma(\vec{u})) = \W(\vec{u})$. In other words, it treats the buyers equally.
    \item (Non-decreasing) For any $\vec{u_1} \leq \vec{u_2}$, $\W(\vec{u_1}) \leq \W(\vec{u_2})$. In other words, it weakly prefers Pareto improvements.
    \item (Concavity) $\W$ is concave. In other words, it weakly prefers a balanced allocation to a convex combination of extremes with the same expected allocation.
    \item (Normalization) $\W(\vec{0}) = 0$. (It also suffices to alternatively require non-negativity: $\W(\vec{u}) \geq 0$ for any $\vec{u} \geq 0$.)
\end{itemize}

This definition captures many common welfare functions, such as the utilitarian social welfare function that outputs the sum, the Nash welfare function that outputs the geometric mean, and the max-min (a.k.a. egalitarian) welfare function that outputs the minimum. It will be clear that we cannot hope for similar results if we drop any of these four conditions.

We will show the surprising existence of a universal scheme -- we do not need to know $\W$ in order to approximately optimize it. Our technical tool to deal with the unknown $\W$ is \emph{majorization}.\footnote{More accurately, we use the notion of \emph{majorization from above}, a.k.a. \emph{supermajorization}. For simplicity, we use the term \emph{majorization} throughout this paper.} Below we define it with the related notions which we will need later.

\paragraph{Majorization.} Given a signaling scheme $\nz$, we define its \emph{surplus-mass function} to be a step function over $(0, 1]$ taking value $\cs_{v_i}(\nz)$ on the interval $\left(F_{\nd}(v_{i-1}),F_{\nd}(v_i)\right]$. Formally, we have:
\begin{definition}[Surplus-Mass Function] \label{def:s_function}
Given a signaling scheme $\nz$, the \emph{surplus-mass function} induced from $\nz$ is a step function $s_\nz:\left(0, 1\right]\rightarrow \mathbb{R}_{\ge 0}$ that satisfies $\forall\,x\in(0,1]$ and $i\in [n]$:
\[
s_\nz(x)=\cs_{v_i}(\nz),\ \forall x\in \left(F_{\nd}(v_{i-1}), F_{\nd}(v_i)\right].
\]
\end{definition}
That is, the surplus-mass function maps a quantile in the value distribution to the expected surplus of the buyer with that value.

\begin{definition}[Integration Prefix Sum]\label{def:vprefix}
Given a function $f:(0,1] \rightarrow \mathbb{R}_{\ge 0}$ and $m \in (0,1]$. The \emph{$m$-integration prefix sum} of $f$ is
\[
\pfv(f,m)=\int_{0}^m f(x)\d x.
\]
\end{definition}

Next, we define the \emph{sorted $m$-prefix sum} of any step function $f$ as the area under the curve over the leftmost $m$-length interval of the ``sorted function'' obtained by sorting the segments of $f$.

\begin{definition}[Sorted Prefix Sum]
Given a step function $f:(0,1]\rightarrow \mathbb{R}_{\ge0}$ with finite steps and a real number $m\in (0,1]$, define a new \emph{sorted} function $f_{\mathrm{sorted}}(x)$ by rearranging the segments in $f$ in the ascending order of $f(x)$ (while keeping the domain $(0, 1]$ unchanged). The \emph{sorted $m$-prefix sum} of $f$ is
\[
\pf(f,m)=\int_{0}^m f_{\mathrm{sorted}}(x)\d x.
\]
\end{definition}
In other words, the sorted $m$-prefix sum outputs the minimum possible (over $S$) integral of $f(x)$ on $x \in S$, where $S \subseteq [0, 1]$ is a finite union of disjoint intervals with total length of $m$.

We now define the majorization relation between two signaling schemes as follows: 
\begin{definition} [Majorization Relation]
A signaling scheme $\nz_1$ is \emph{majorized} by another signaling scheme $\nz_2$ if
\[
\forall\, m\in \left(0, 1\right],\ \pf(s_{\nz_1},m)\ge \pf(s_{\nz_2},m),
\]
where $s_{\nz_1},s_{\nz_2}$ are the surplus-mass functions induced under schemes $\nz_1$ and $\nz_2$ respectively. A signaling scheme is said to be {\em majorized} if it is majorized by every other signaling scheme.
\end{definition}

\begin{example}
\label{ex:majorization}
    In our running example, the expected consumer surplus  under $\nzeo$ (resp. $\nzet$) is $\langle 0, 0.6, 0.4, 3 \rangle$ (resp. $\langle 0, \frac{1}{7}, \frac{10}{7}, \frac{17}{7}\rangle$). Each of these surplus values occupies mass of $1/4$. Thus,  
    \begin{align*}    
    \pf\Big(\nzet, \frac{1}{2}\Big) = \frac{\sum_{v\in\{1,2\}}\cs_v(\nzet)}{4}  = \frac{1}{28} &< \frac{1}{10} = \frac{\sum_{v\in\{1,3\}}\cs_v(\nzeo)}{4}  = \pf\big(\nzeo, \frac{1}{2}\big);\\
    \pf\Big(\nzeo, \frac{3}{4}\Big) = \frac{\sum_{v\in\{1,2,3\}}\cs_v(\nzeo)}{4}=\frac{1}{4} &< \frac{11}{28} = \frac{\sum_{v\in\{1,2,3\}}\cs_v(\nzet)}{4} = \pf\Big(\nzet, \frac{3}{4}\Big).
    \end{align*}
    Thus $\nzet$ is not majorized by $\nzeo$ and $\nzeo$ is not majorized by $\nzet$, and hence neither $\nzeo$ nor $\nzet$ can be majorized by every other signaling scheme.
\end{example}

Indeed, in \cref{sec:lb}, we show that (exact) majorization is unattainable -- there are instances where no signaling scheme is majorized by every other signaling scheme.  Given this, we define the following approximation version of majorization. 

\begin{definition} [$\alpha$-Majorization]
\label{def:approx_maj}
A signaling scheme $\nz_1$ is \emph{$\alpha$-majorized} by another signaling scheme $\nz_2$ if $\forall\, m\in \left(0, 1\right]$, we have:
\[
\alpha\cdot \pf(s_{\nz_1},m)\ge \pf(s_{\nz_2}, m).
\]
Further, we say a signaling scheme $\nz$ is \emph{$\alpha$-majorized} if it is $\alpha$-majorized by every other signaling scheme $\nz'$.
\end{definition}

The following established fact~\citep{hardy1952inequalities,GoelM06} is crucial in our universal maximization of well-behaved welfare functions. We include a proof for completeness.
\begin{proposition} [Proved in \cref{app:omitted}]
\label{prop:majorize_concave}
Any $\alpha$-majorized signaling scheme $\nz$ gives an $\alpha$-approximation to the welfare under any signaling scheme, as long as the welfare function is symmetric, weakly increasing, concave, and normalized (or non-negative). Conversely, if a signaling scheme $\nz$ gives an $\alpha$-approximation to all such welfare functions, it must be $\alpha$-majorized.
\end{proposition}

\section{Finding an $8$-Majorized Signaling Scheme}

In this section, we construct an 8-majorized signaling scheme. 
In more detail, in~\cref{sec:vprefix}, we present our \algname algorithm (\cref{alg:match}) that given any empirical distribution $\nd$ constructs a signaling scheme $\nzz$ that approximates the $m$-integration prefix sum of any other signaling scheme:
\begin{lemma} \label{thm:value_pr}
Given $\nd$, let $\nzz$ denote the signaling scheme returned by \algname Algorithm (\cref{alg:match}). Then, for any signaling scheme $\nz'$ and any $m \in (0,1]$, we have
\[
4 \cdot \pfv(s_{\nzz}, m) \ge \pfv(s_{\nz'}, m).
\]
\end{lemma}

However, in order to achieve $\alpha$-majorization, we need to approximate the optimal \emph{sorted} prefix sum (rather than the optimal integration prefix sum). In \cref{sec:iron}, we show an ironing process that transforms the surplus-mass function $s_{\nzz}(x)$ into a monotonically increasing step function $\ts(x)$, while preserving the integration prefix sum of $s_{\nzz}$ at any point of discontinuity. Based on $\nzz$, we then construct a monotone signaling scheme $\nzf$ such that the surplus-mass function of $\nzf$ is exactly half of $\ts$:

\begin{lemma} \label{thm:smooth}
$\forall x \in (0,1], s_{\nzf}(x)=\frac{1}{2}\cdot \ts(x).$
\end{lemma}

Combining~\cref{thm:value_pr,thm:smooth} leads to our main result:
\begin{theorem}\label{thm:main}
$\nzf$ is efficient, monotone, and 8-majorized by any other signaling scheme. 
\end{theorem}

\subsection{Construction of $\nzz$}\label{sec:vprefix}

In this section, we construct the signaling scheme $\nzz$. 
The main idea is to decompose any given $\nd$ into only two types of posterior distributions, which we refer to as \emph{singleton} and \emph{equal-revenue binary} signals.

\begin{definition}[Singleton Signal]
A signal $\ns=\nss{v_i}$ is said to be a \emph{singleton signal} on $v_i$ if it satisfies $f_{\ns}(v) = \mathds{1}[v = v_i]$.
\end{definition}

\begin{definition}[Equal-Revenue Binary Signal]\label{def:bin_sig}
A signal $\ns=\ns_{v_i, v_j}^{E}$ is said to be an \emph{equal-revenue binary signal} on $v_i < v_j$ if it satisfies:
\begin{equation*}
f_{\ns}(v)=\begin{cases}
1-\frac{v_i}{v_j} & v=v_i;\\
\frac{v_i}{v_j} & v=v_j;\\
0 & v\notin \{v_i,v_j\}.
\end{cases}
\end{equation*}
\end{definition}
Note that if the seller receives $\ns_{v_i, v_j}^{E}$, then posting a price of either $v_i$ and $v_j$ leads to the same revenue (hence ``equal-revenue'').
We assume that the equal-revenue binary signals are indexed from $1$ to $Q_1$, where for each $q \in [Q_1]$, we have the signal $\ns_q = \ns_{v_q^1, v_q^2}^\mathrm{E}$. We call the higher value $v_q^2$ \emph{taker} and the lower value $v_q^1$ \emph{giver}. The masses on them (i.e. $\gamma_q  f_{\ns_q}(v_q^2)$ and $\gamma_q f_{\ns_q}(v_q^1)$ respectively) are called \emph{taker mass} and \emph{giver mass} respectively. 
Using these definitions, we can describe our \algname Algorithm in \cref{alg:match}.

\begin{algorithm}[h]
\SetKwComment{Comment}{[}{]}
\SetKwInput{KwData}{Input}
\SetKwInput{KwResult}{Output}
\caption{\algname Algorithm}\label{alg:match}
\KwData{Distribution $\nd=\{v_i,f_{\nd}(v_i)\}_{i\in [n]}$.}
\KwResult{A signaling scheme $\nzz=\left\{\ns_q, \gamma_q\right\}_{q\in [Q]}$.}
\textbf{Initialize:} $\nzz \gets \varnothing;$ $q\gets 0; (\mg_i,\mt_i)\gets \left(\frac{1}{2}\cdot f_{\nd}(v_i),\frac{1}{2}\cdot f_{\nd}(v_i)\right)\,\forall\,i\in[n]$\;
\Repeat{no such $(s, l)$ exists}{
    $q\gets q+1$\;
    Find the smallest $s\in [n]$ such that $\mg_s>0$\;
    Find the smallest $\ell > s$ such that $\mt_\ell>0$\label{alg_line:find_ell}\;
    Set $\ns_q = \ns^\mathrm{E}_{v_s,v_\ell}$ and $\gamma_q =\min\{\mg_s\big/\big(1-\frac{v_s}{v_\ell}\big), \mt_\ell \big/\big(\frac{v_s}{v_\ell}\big)\}$\;
    Add signal $(\ns_q,\gamma_q)$ to $\nzz$\;
    Update $\mg_s\gets \mg_s-\gamma_q\cdot(1-\frac{v_s}{v_\ell});\ \mt_\ell\gets \mt_\ell-\gamma_q\cdot\frac{v_s}{v_\ell}$\;
}
Cover all remaining masses using singleton signals and add them to $\nzz$.  
\end{algorithm}

To understand this construction, first note that by the Bayes plausibility of a signaling scheme, we have the following set of linear constraints on any equal-revenue binary signal:
\begin{equation} \label{eq:balance}
    \forall\, i \in [n], \ \sum_{q\in Q_1: v_q^1=v_i} \big(1 - \frac{v_i}{v_j}\big) \cdot \gamma_q + \sum_{q \in Q_1: v_q^2 = v_i} \big(\frac{v_i}{v_j}\big) \cdot \gamma_q \le f_{\nd}(v_i).
\end{equation}
In the construction of $\nzz$, we strengthen these constraints into the following:
\begin{gather}
    \forall\, i \in [n], \ \sum_{q\in Q_1: v_q^1=v_i} \big(1 - \frac{v_i}{v_j}\big) \cdot \gamma_q \le \frac{1}{2} \cdot f_{\nd}(v_i);\label{eqn:bernoulli_1}\\
    \forall\, i \in [n], \ \sum_{q \in Q_1: v_q^2 = v_i} \big(\frac{v_i}{v_j}\big) \cdot \gamma_q \le \frac{1}{2} \cdot f_{\nd}(v_i).\label{eqn:bernoulli_2}
\end{gather}
We conduct a greedy process to find a solution satisfying the strengthened constraints. We iteratively find the smallest index $s$ such that \cref{eqn:bernoulli_1} is slack for $s$ and the smallest index $\ell\, (\ell > s)$ such that \cref{eqn:bernoulli_2} is slack for $\ell$. We add a maximal equal-revenue binary signal with supports $v_{s}$ and $v_{\ell}$ so that one of the two constraints becomes tight. We iterate until no such pair of $(i_s, i_\ell)$ exists. \cref{fig:matching} illustrates the construction of our signaling scheme $\nzz$ obtained by running \cref{alg:match} on some given $\nd$.


\begin{figure}[htbp]
    \centering  \includegraphics{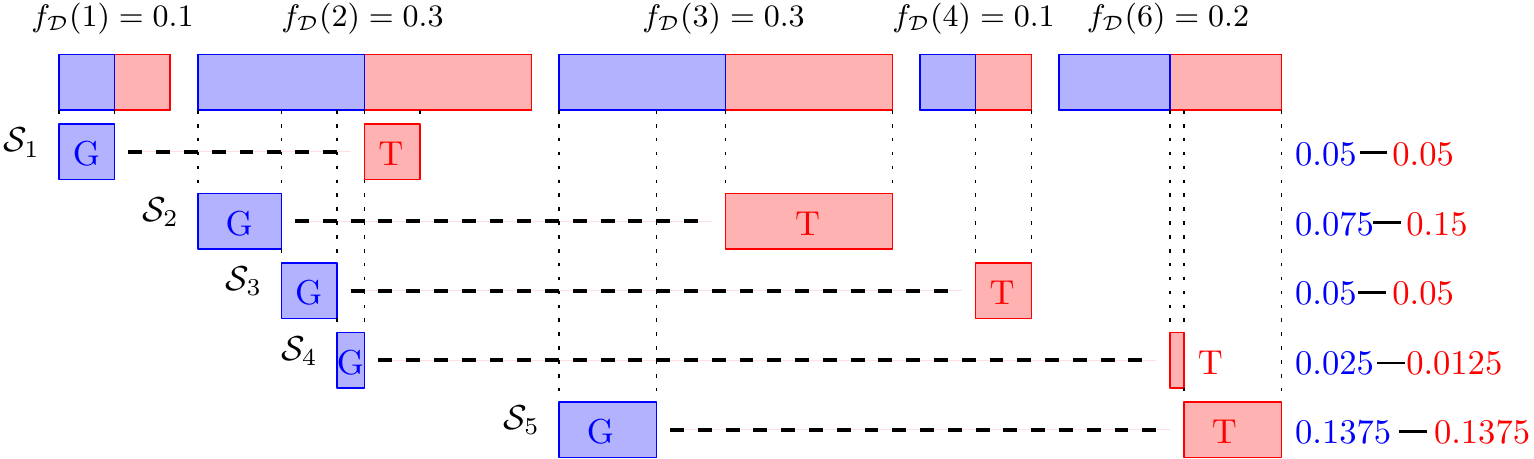}
    \caption{\small\it Illustrating the construction of $\nzz$ using \cref{alg:match}: The buyer values are $\langle 1,2,3,4,6\rangle$ with distribution $\nd=\langle 0.1, 0.3, 0.3, 0.1, 0.2\rangle$, indicated using scaled rectangles in the first row. Each subsequent row corresponds to an equal-revenue binary signal $\ns_q=\ns_{v_i, v_j}^{E}$ with weight $\gamma_q$. The letters $\mathrm{G}$ and $\mathrm{T}$ together with the blue and red rectangles represent the giver and the taker corresponding to each signal, while the blue and red numbers on the far right of each row are the giver and taker masses ($\gamma_qf_{\ns_q}(v_i),\gamma_qf_{\ns_q}(v_j)$); these are also illustrated by the lengths of the rectangles in that row. (Again, we can recover each signal and its weight by normalizing the numbers. For example, signal $\ns_1$ satisfies $f_{\ns_1}(1)=f_{\ns_1}(2)=0.5$ and has weight $\gamma_1=0.1$.)
}
    \label{fig:matching}
\end{figure}
\subsection{Approximating $\pfv$ via \algname  (Proof of \cref{thm:value_pr})}

Recall for any $k$, $\mks=\sum_{i=1}^k f_{\nd}(v_i)$ denotes the total population of buyers with value at most $v_k$. 
We now prove \cref{thm:value_pr} in two steps: First, in \cref{sec:upper_truncate}, we show an upper bound on $\pfv\big(s_{\nz'}, \mks\big)$ for any signal $\nz'$. This generalizes the corresponding bound of \citet*{Bergemann15} to a sub-population. Next, in \cref{sec:proof_of_lowerbound}, we show that the $\pfv$ values of $\nzz$ approximately achieve this upper bound.

\subsubsection{Bounding $\pfv$ via the Surplus of Truncated Distributions}\label{sec:upper_truncate}

As a thought experiment, we restrict our attention to the subset of buyers with the smallest $k$ values. What is the maximum possible consumer surplus on this sub-population? In \cref{lem:pf_bound}, we show that it is upper bounded by the total values in the sub-population, minus the revenue extractable from this sub-population without signaling. We need the following definitions to present the proof.
\begin{definition}[Truncated Distribution]
Given distribution $\nd$ with finite support $S\subset \mathbb{R}_{\ge 0}$ and any $x\in S$, the truncated distribution of $\nd$ on $x$, denoted by $\overline{\nd}(x)$, satisfies:
\begin{equation*}
f_{\overline{\nd}(x)}(v)=\begin{cases}
\frac{f_{\nd}(v)}{F_{\nd}(x)} & v\in S \mbox{ and } v\le x;\\
0 & \mbox{Otherwise}.
\end{cases}
\end{equation*}
\end{definition}

\begin{definition}[Surplus Prefix Sum]
\label{def:welfare_prefix}
Given any $k\in[n]$, the \emph{$k$-surplus prefix sum} of the buyers is 
\[
V_k=\sum_{i=1}^k v_i\cdot f_{\nd}(v_i).
\]
\end{definition}




\begin{lemma}\label{lem:pf_bound}
For any $k\in [n]$ and any signaling scheme $\nz'$, we have $$
\pfv\big(s_{\nz'}, \mks\big)\le V_k-\max_{i\in[k]}\left\{v_{i}\cdot \sum_{j=i}^k f_{\nd}(v_{j})\right\}.
$$
\end{lemma}

\begin{proof}
Suppose that $\nz'=\{(\ns_q, \gamma_q)\}_{q\in[Q]}$. We have 
\begin{align*}
\pfv\big(s_{\nz'}, \mks\big) 
&=\sum_{i=1}^k f_{\nd}(v_i)\cdot \left(\sum_{q=1}^Q \frac{\gamma_q\cdot f_{\ns_q}(v_i)}{f_{\nd}(v_i)}\cdot \mathds{1}[v_i\ge p_{\ns_q}^*]\cdot (v_i-p_{\ns_q}^*)\right).
\end{align*}


Let $\hat{\ns}_q=\overline{\ns_q}(v_k)$ be the truncated distribution on the buyers with value at most $v_k$. Denote the optimal price of $\hat{\ns}_q$ by ${p}_{\hat{\ns}_q}^*$. We now claim that ${p}_{\hat{\ns}_q}^*\le p_{\ns_q}^*$. First assume $p_{\ns_q}^* \le v_k$, otherwise the statement is trivial. Next, to find the optimal price, we can ignore the scaling factor $1/F_{\nd}(v_k)$ (since it scales up the revenue for each price by the same amount).
Thus, when we truncate at $v_k$, we can view it as removing some probability mass $\mu_k$ beyond $v_k$.  Note that for any price $p$, the decrease in revenue is $p \cdot \mu_k$. This means larger prices suffer larger drops in revenue, i.e., the new optimal price ${p}_{\hat{\ns}_q}^*$ cannot be larger. 



Moreover, since $(v_i - p)$ is monotonically decreasing as a function of $p$, substituting ${p}_{\hat{\ns}_q}^*$ for $p_{\ns_q}^*$ in the above equality, we have: 
\begin{align*}
\pfv\big(s_{\nz'}, \mks\big)&  \le
 \sum_{q=1}^Q \left(\sum_{i\in[k]:v_i\ge p_{\hat{\ns}_q}^*}  \gamma_q\cdot f_{\ns_q}(v_i)\cdot (v_i-p_{\hat{\ns}_q}^*)\right) \\
  &\le \sum_{q=1}^Q \left( \sum_{i\in [k]}\gamma_q\cdot f_{\ns_q}(v_i)\cdot v_i\right)-\sum_{q=1}^Q \gamma_q\cdot \left(\sum_{i\in[k]:v_i\ge p_{\hat{\ns}_q}^*} f_{\ns_q}(v_i)\cdot p_{\hat{\ns}_q}^*\right).\\
\end{align*}
Let $i_k^*=\argmax_{i\in[k]}\left\{v_{i}\cdot \sum_{i'=i}^k f_{\nd}(v_{i'})\right\}$. Since $p_{\hat{\ns}_q}^*$ is the optimal price on $\hat{S}_q$, we have 
$$p_{\hat{\ns}_q^*} \sum_{i\in[k]:v_i\ge p_{\hat{\ns}_q}^*} f_{\ns_q}(v_i)\ge v_{i_k^*} \sum_{i\in[k]:v_i\ge v_{i_k^*}} f_{\ns_q}(v_i)$$ 
for any $q\in [Q]$. Plugging this into the previous inequality, we finally have:
\begin{align*}
 \pfv\big(s_{\nz'}, \mks\big)&\le 
  \sum_{i=1}^k \left( v_i\cdot\sum_{q=1}^Q   \gamma_q\cdot f_{\ns_q}(v_i)\right)-\sum_{q=1}^Q \gamma_q\cdot \left(\sum_{i\in[k]:v_i\ge v_{i_k^*}} f_{\ns_q}(v_i)\cdot v_{i_k^*}\right) \\
 &= \sum_{i=1}^{k} v_i\cdot f_{\nd}(v_i)-v_{i_k^*}\cdot \sum_{i =i_k^*}^k f_{\nd}(v_i)\\
 &= V_k-\max_{i\in[k]}\left\{v_{i}\cdot \sum_{i'=i}^k f_{\nd}(v_{i'})\right\}. \qedhere
\end{align*}
\end{proof}

\subsubsection{Approximating Prefix Sums via the \algname Algorithm}
\label{sec:proof_of_lowerbound}
By \cref{def:s_function}, for any signaling scheme $\nz$ and any $k\in [n]$, $s_{\nz}$ is constant on the interval $(F_{\nd}(v_{k-1}), F_{\nd}(v_{k})]$. Therefore, $\pfv(\nz, m)$ is a linear function of $m$ on the interval $(F_{\nd}(v_{k-1}), F_{\nd}(v_{k})]$. Therefore, to prove \cref{thm:value_pr}, it suffices to show it when $m=\mks$ for $k\in [n]$. Moreover, we can further replace $\pfv(s_{\nz'},\mks)$ with the upper bound we obtained \cref{lem:pf_bound} (i.e, with the maximum consumer surplus of truncated distributions). Thus we can obtain \cref{thm:value_pr} as an immediate consequence of the following lemma.

\begin{lemma}\label{thm:value_p}
Let $\nzz$ be the signaling scheme returned by~\cref{alg:match} for a given $\nd$. Then for any signaling scheme $\nz'$ and any $k\in [n]$, we have 
$$
4\cdot \pfv(s_{\nzz}, \mks)\ge V_k-\max_{i\in[k]}\left\{v_{i}\cdot \sum_{i'=i}^k f_{\nd}(v_{i'})\right\}.
$$
\end{lemma}

Consider the first point in time in \cref{alg:match} when $\ell=k+1$; if $k = n$, this is the stopping time of the algorithm. Let $i^*$ be the smallest index such that \cref{eqn:bernoulli_1} is still slack at this point in time. This means \cref{eqn:bernoulli_1} is tight for all $i\in [i^*-1]$ and \cref{eqn:bernoulli_2} is tight for all $i\in [i^*+1, k]$ at this point in time. Note that $i^*=\min\{i \mid \mg_i > 0\}$.

We now prove two lower bounds on $\pfv(s_{\nzz}, \mks)$. Recall the definition of $V_k$ from \cref{def:welfare_prefix}. We have:

\begin{proposition}\label{prop:double_le_vpf} 
 \label{prop:double_le_rest} We have the following inequalities:
\begin{itemize}
\item  $2 \cdot \pfv(s_{\nzz}, \mks) \ge V_{i^*-1}$.
\item  $2 \cdot \pfv(s_{\nzz}, \mks) \ge V_{k}-V_{i^*-1}-\max_{i\in[k]}\left\{v_{i}\cdot \sum_{i'=i}^k f_{\nd}(v_{i'})\right\}$.
\end{itemize}
\end{proposition}
\begin{proof} 
Assume that the $j^{\text{th}}$ equal-revenue binary signal added to $\nzz$ during \cref{alg:match} is $\ns_j=\ns_{\vg_j,\vt_j}^\mathrm{E}$ with weight $\gamma_j$. Let $\tg_j=f_{\ns_j}(\vg_j)\cdot \gamma_j$ and $\tt_j=f_{\ns_j}(\vt_j)\cdot \gamma_j$. By \cref{def:bin_sig}, we have 
\begin{equation}
\tt_j\cdot \vt_j=(\tg_j+\tt_j)\cdot \vg_j.\label{eqn:equal}
\end{equation}
Therefore,
\begin{align*}
2\cdot \pfv(s_{\nzz}, \mks)&=2\cdot \sum_{i=1}^k \left(\sum_{j:\, \vt_j=v_i} (\vt_j-\vg_j)\cdot \tt_j\right)\tag{By definition of buyers' surplus}\\
&= 2\cdot \sum_{i=1}^k \left(\sum_{j:\, \vt_j=v_i} \vg_j\cdot \tg_j\right). \tag{By \cref{eqn:equal}}
\end{align*}
Since \cref{eqn:bernoulli_1} is tight for any $i<i^*$, we have 
 $$\forall\, i \in [i^*-1],\ \sum_{j:\,\vg_j=v_i,\vt_j\le v_k} \tg_j = \frac{1}{2}\cdot f_{\nd}(v_i).$$
Since $\sum_{i=1}^k \left(\sum_{j:\, \vt_j=v_i} \vg_j\cdot \tg_j\right)\ge \sum_{i=1}^{i^*-1}\left(v_i\cdot \sum_{j:\,\vg_j=v_i, \vt_j\le v_k} \tg_j\right)$, we further have 
\begin{equation}
    2\cdot \pfv(s_{\nzz}, \mks)\ge \sum_{i=1}^{i^*-1}  v_i \cdot f_{\nd}(v_i)=V_{i^*-1}.
\end{equation}
This completes the proof of the first inequality. To show the second inequality, we have
\begin{align*}
2\cdot \pfv(s_{\nzz}, \mks)&\ge 2\cdot \sum_{i=i^*+1}^{k}\Bigg(\sum_{j:\,\vt_j=v_i} (\vt_j-\vg_j)\cdot \tt_j\Bigg)\\
&\ge 2\cdot \sum_{i=i^*+1}^{k}\Bigg(\sum_{j:\,\vt_j=v_i} (\vt_j-v_{i^*})\cdot \tt_j\Bigg)\tag{Since all the giver values are at most $v_{i^*}$}\\
&= 2\cdot \sum_{i=i^*+1}^{k}\Bigg(\sum_{j:\,\vt_j=v_i}\vt_j\cdot \tt_j\Bigg)-2\cdot v_{i^*}\cdot \sum_{i=i^*+1}^{k}\Bigg(\sum_{j:\,\vt_j=v_i} \tt_j\Bigg).
\end{align*}
Since \cref{eqn:bernoulli_2} is tight for any $i\in [i^*+1,\ k]$, this means 
\[
\forall i \in [i^*+1,\ k],\ 
\sum_{j:\, \vt_j=v_i} \tt_j=\frac{1}{2}\cdot f_{\nd}(v_i).
\]
Using this in the above derivation, we have:
\begin{align*}
2\cdot \pfv(s_{\nzz}, \mks)&\geq \sum_{i=i^*+1}^{k} f_{\nd}(v_i)\cdot v_i-  v_{i^*}\cdot \sum_{i=i^*+1}^{k} f_{\nd}(v_i)\\
&= \sum_{i=i^*}^{k}f_{\nd}(v_i)\cdot v_i-v_{i^*}\cdot \sum_{i=i^*}^{k} f_{\nd}(v_i)\\
&\ge V_k-V_{i^*-1}-\max_{i\in[k]}\left\{v_{i}\cdot \sum_{i'=i}^k f_{\nd}(v_{i'})\right\}.
\end{align*}
This completes the proof of the second inequality.
\end{proof}

\begin{proof}[Proof of \cref{thm:value_p}]
Adding the inequalities in the proposition above, we have 
\begin{align*}
4\cdot \pfv(s_{\nzz}, \mks)&\ge V_{i^*-1} +V_k-V_{i^*-1}-v_{i^*}\cdot \sum_{i=i^*}^{k} f_{\nd}(v_i)\\
&\ge V_{k}-\max_{i\in [k]}v_{i}\cdot \sum_{i'=i}^{k} f_{\nd}(v_{i'})\\
&\ge \pfv(s_{\nz'},\mks),\tag{By \cref{lem:pf_bound}}
\end{align*}
completing the proof of \cref{thm:value_p} and hence that of \cref{thm:value_pr}.
\end{proof}

\subsection{Extending to Sorted Prefix Sums via Ironing and Smoothing}\label{sec:iron}

We will now prove \cref{thm:smooth}. In \cref{sec:ironing}, we introduce the (classical) ironing process that transforms the surplus-mass function $s_{\nzz}(x)$ into $\ts(x)$. After that, in \cref{sec:smoothing}, we conduct a smoothing process to obtain a signaling scheme $\nzm$ whose induced surplus-mass function is at least half of $\ts(x)$ (see \cref{lem:smooth}). The ironing process can be intuitively viewed as moving surplus from some high-surplus but lower-value buyers to some high-value but low-surplus buyers, so that the ironed function is monotone. It makes the surplus-mass function more ``even''. The smoothing process describes what specific modification we should operate on the signaling scheme (or signals) to achieve the ironing purpose on the surplus-mass function. This will show \cref{thm:smooth}.
Denote the final signaling scheme after decomposition by $\nzf$. In \cref{sec:proof_of_main}, we show that $\nzf$ simultaneously guarantees 8-majorization and monotonicity, thus completing the proof of \cref{thm:main}.

\subsubsection{Ironing}\label{sec:ironing}
So far, we have approximated the optimal integration prefix sum of consumer surplus. We need to transform the approximation on integration prefix sum into the approximation on {\em sorted} prefix sum, which will yield the bound on approximate majorization. Our first step is to process the surplus-mass function via ironing. Ironing is a standard process on functions to achieve monotonicity. It is first applied in the auction scenario by \citet{Myerson81}. For the completeness of our paper, we also include a description of the ironing process in this section.

Consider the surplus-mass function $s_{\nzz}(x)$. We operate the ironing process as follows: 
\begin{itemize}
\item[1.] Compute the integral of $s_{\nzz}(x)$ as $F(x)$.
\item[2.] Compute the lower convex envelope of $F(x)$, denoted as $\tilde{F}(x)$.
\item[3.] Compute the derivative of $\tilde{F}(x)$ as the ironed function $\ts(x)$. Define the value at any point of discontinuity as its left limit.
\end{itemize}
We have the following properties of $\ts(x)$. 
\begin{lemma}\label{lem:ts_prop}
    The ironed function $\ts(x)$ satisfies:
    \begin{itemize}
        \item [1.] $\ts(x)$ is weakly increasing;
        \item [2.] For any $m\in (0, 1]$, we have $\pfv(\ts, m)\le \pfv(s_{\nzz},m)$.
    \end{itemize}
\end{lemma}
\begin{proof}
    Since $\tilde{F}$ is convex, we have its derivative $\ts$ is weakly increasing. Moreover, since the convex envelope has property that $\tilde{F}(x) \le F(x)$ and the $\pfv$ is defined by the integration from 0, we have $\forall\, m\in (0, 1]$, $\pfv(\ts, m)\le \pfv(s_{\nzz},m)$.
\end{proof}

Since both $\ts(x)$ and $s_{\nzz}(x)$ are step functions, the range of $x$ where $F(x)$ and $\tilde{F}(x)$ are different consists of a collection of open intervals. On the graph depicting $F(x)$ and $\tilde{F}(x)$, each interval represents a region of $x$ such that $\tilde{F}(x)$ falls below $F(x)$.

We call these open intervals ``ironing intervals'' and denote them by $I_1,I_2,\ldots, I_T$. Within each interval $I_t$, $\tilde{F}(x)$ is a linear function, and thus $\ts(x)$ is a constant. We denote this constant by $\ts_t$.

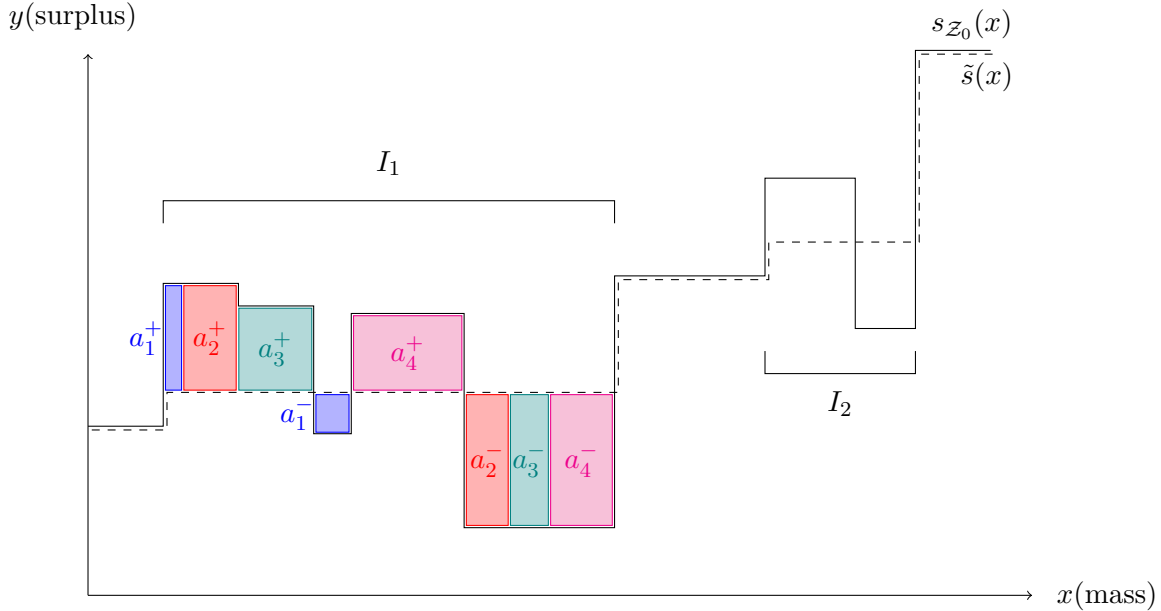
\begin{figure}[htbp]
\centering
\begin{tikzpicture}[scale=1.0]
\draw [dashed] (0,0)--(1,0)--(1,0.5)--(7,0.5)--(7,2)--(9,2)--(9,2.5)--(11,2.5)--(11,5)--(12,5);
\draw [->] (-0.05, -2.2)--(12.5,-2.2);
\node at (13.5, -2.2) {$x$(mass)};
\draw [->] (-0.05, -2.2)--(-0.05, 5);
\node at (-0.25, 5.5) {$y$(surplus)};
\begin{scope}[shift={(-0.05,0.05)}]
\draw (0, 0)--(1, 0)--(1,1.9)--(2,1.9)--(2, 1.6)--(3,1.6)--(3,-0.1)--(3.5,-0.1)--(3.5,1.5)--(5,1.5)--(5,-1.35)--(7,-1.35)--(7,2)--(9,2)--(9,3.3)--(10.2,3.3)--(10.2, 1.3)--(11, 1.3)--(11,5)--(12,5);
\draw (1,2.7)--(1,3)--(7,3)--(7,2.7);
\draw (4,3.5) node {$I_1$}; 
\draw (9,1)--(9,0.7)--(11,0.7)--(11,1);
\draw (10,0.3) node {$I_2$}; 
\end{scope}
\node at (11.7, 5.4) {$s_{\nzz}(x)$};
\node at (11.9, 4.7) {$\ts(x)$};
\begin{scope}[shift={(-0.02,0.02)}]
\draw[draw=blue, fill=blue!30] (1,0.51) rectangle ++(0.214,1.39) node [left,midway, text=blue] {$a_1^+$};
\draw[draw=red, fill=red!30] (1.244,0.51) rectangle ++(0.696,1.39) node [pos=.5, text=red] {$a_2^+$};
\draw[draw=teal, fill=teal!30] (1.975,0.51) rectangle ++(0.97,1.09) node [pos=.5, text=teal] {$a_3^+$};
\draw[draw=blue, fill=blue!30] (3,-0.05) rectangle ++(0.44,0.5) node [left=3, midway, text=blue] {$a_1^-$};
\draw[draw=magenta, fill=magenta!30] (3.50,0.51) rectangle ++(1.44,0.99) node [pos=.5, text=magenta] {$a_4^+$};
\draw[draw=red, fill=red!30] (5, -1.29) rectangle ++(0.556,1.74) node [pos=.5, text=red] {$a_2^-$};
\draw[draw=teal, fill=teal!30] (5.586, -1.29) rectangle ++(0.504,1.74) node [pos=.5, text=teal] {$a_3^-$};
\draw[draw=magenta, fill=magenta!30] (6.12, -1.29) rectangle ++(0.82,1.74) node [pos=.5, text=magenta] {$a_4^-$};
\end{scope}

\end{tikzpicture}
\caption{The solid line denotes the function $s_{\nzz}(x)$. After ironing, the dashed line denotes the function $\ts(x)$. The masses are sorted in ascending order of the values. The four pairs of rectangles with the same colour share the same area respectively. The first two ironing intervals are shown as $I_1$ and $I_2$. }
\label{fig:iron}
\end{figure}

Fix an ironing interval $I_t=(\ell_t, r_t)$. By the definition of ironing intervals, we have that for any $x_0\in \{\ell_t, r_t\}$, $F(x_0)=\int_{0}^{x_0} s_{\nzz}(x)\d x=\int_{0}^{x_0} \ts(x)\d x=\tilde{F}(x_0)$, thus $\int_{\ell_t}^{r_t} s_{\nzz}(x)\d x= \int_{\ell_t}^{r_t} \ts(x)\d x$. Equivalently, we have $$\int_{x\in I_t: \ts(x)>s_{\nzz}(x)} (\ts(x)-s_{\nzz}(x)) \d x= \int_{x\in I_t: \ts(x)\le s_{\nzz}(x)} (s_{\nzz}(x)-\ts(x)) \d x.$$
This means $\int_{x\in I_t} [s_{\nzz}(x) - \ts_t]^+ \d x = \int_{x\in I_t} [\ts_t - s_{\nzz}(x)]^+ \d x$. Moreover, by the second property in \cref{lem:ts_prop}, we have $$\forall x_0\in (\ell_t, r_t],\ \int_{\ell_t}^{x_0}[s_{\nzz}(x) - \ts_t]^+ \d x\ge \int_{\ell_t}^{x_0}[\ts_t-s_{\nzz}(x)]^+ \d x.$$ 

Based on these two observations, we can split the area above $y=\ts_t$ while below $s_{\nzz}(x)$ into ${Y_t}$ rectangles $\{a_1^+,a_2^+,\ldots,a_{Y_t}^+\}$, as well as the area below $y=\ts_t$ while above $s_{\nzz}(x)$ into the same number of rectangles $\{a_1^-,a_2^-,\ldots, a_{Y_t}^-\}$. \cref{alg:rectangle} describes the process of constructing such pairs of rectangles. For any $y\in [Y_t]$, the pair of rectangles $(a_y^+, a_y^-)$ satisfies the following  conditions: 
\begin{itemize}
    \item [1.] They have the same area;
    \item [2.] $a_y^+$ is on the left of $a_y^-$;
    \item [3.] Each rectangle $a_y^+$ (resp. $a_y^-$) corresponds to a single buyer value $v_y^+$ (resp. $v_y^-$).
\end{itemize}

\Cref{fig:iron} shows an example of ironing and pairing of rectangles within the interval $I_1$. The four rectangles above $y=\ts_t$ (i.e. $\{a_1^+, a_2^+, a_3^+, a_4^+\}$) have the same areas as the four rectangles below $y=\ts_t$ (i.e. $\{a_1^-, a_2^-, a_3^-, a_4^-\}$) respectively. 

\begin{algorithm}[htbp]
\SetKwComment{Comment}{[}{]}
 \SetKwInput{KwData}{Input}
\SetKwInput{KwResult}{Output}
\caption{Construct a rectangle pairing on $I_t$ \label{alg:rectangle}}
\KwData{$s_{\nzz}(x)$, $I_t=(\ell_t,r_t]$, $\ts_t$ (ironed surplus on $I_t$).}
\KwResult{A set of rectangle pairs $\{(a_y^+, a_y^-)\}_{y\in [Y_t]}$.}
Let $v_{i_1^+} < \cdots < v_{i_g^+}$ be the buyer values whose occupied interval in $I_t$ satisfies $s_{\nzz}(x)>\ts_t$\;
Let $v_{i_1^-} < \cdots < v_{i_h^-}$ be the buyer values whose occupied interval in $I_t$ satisfies $s_{\nzz}(x) < \ts_t$\;
$x^+\gets F_\nd(v_{i_1^{+}-1});\ x^-\gets F_\nd(v_{i_1^{-}-1})$\;
$y\gets 0;\ R\gets \varnothing$\;
\Repeat{no such $(\ell, r)$ exists}{
Find the smallest $\ell \in [g]$ such that $F_\nd(v_{i_\ell^{+}})>x^+;\ x^+ = \max\{x^+, F_\nd(v_{i_\ell^{+} - 1})\}$\;
Find the smallest $r \in [h]$ such that $F_\nd(v_{i_r^{-}})>x^-;\ x^- = \max\{x^-, F_\nd(v_{i_r^{-} - 1})\}$\;
$A\gets \min\left\{(F_\nd(v_{i_\ell^{+}}) - x^+)\cdot (s_{\nzz}(v_{i_\ell^{+}}) - \ts_t),\ (F_\nd(v_{i_r^{-}}) - x^-)\cdot (\ts_t - s_{\nzz}(v_{i_r^{-}}))\right\}$\;
$y\gets y+1$\; 
Let $a_y^+$ be the rectangle with bottom-left corner coordinate at $(x^+, \ts_t)$ with width $w_y^+ = \frac{A}{v_{i_\ell^{+}} - \ts_t}$ and height $h_y^+ = v_{i_\ell^{+}} - \ts_t$\;
Let $a_y^-$ be the rectangle with top-left corner coordinate at $(x^-, \ts_t)$ with width $w_y^- = \frac{A}{\ts_t-v_{i_r^{-}}}$ and height $h_y^- = \ts_t - v_{i_r^{-}}$\;
Add the rectangle pair $(a_y^+, a_y^-)$ to $R$\;
$x^+\gets x^+ + w_y^+;\ x^-\gets x^- + w_y^-$\;
}
\end{algorithm}

\subsubsection{Smoothing} \label{sec:smoothing}
We now present the smoothing process. Consider the ironed interval $I_t$. Our goal is to make all buyers with expected consumer surplus less than $\ts_t/2$ in $\nzz$ (call them \emph{poor} buyers) have expected consumer surplus at least $\ts_t/2$ after smoothing. 
We do so by collecting a portion of giver masses from all the equal-revenue binary signals that contribute to the surplus of carefully chosen high-surplus (or \emph{rich}) buyers. We use the collected masses as giver masses of these signals to construct new equal-revenue binary signals with the poor buyers' values as the taker value, hence bringing their expected surplus to at least $\ts_t/2$. 

We now describe the process in more detail. By applying \cref{alg:rectangle}, we obtain a set of pairs of rectangles $\{(a_y^+, a_y^-)\}_{y\in [Y_t]}$. Suppose that the $x$-coordinates of $a_y^+$ and $a_y^-$ correspond to buyer values $v_y^+$ and $v_y^-$ respectively. Denote the widths of $a_y^+$ and $a_y^-$ by $w_y^+$ and $w_y^-$, and the heights by $h_y^+$ and $h_y^-$ respectively. Note that for $a_y^-$, the expected consumer surplus is $\ts_t - h_y^-$, while for $a_y^+$ it is $\ts_t + h_y^+$.

For each pair of rectangles $(a_y^+,a_y^-)$, if $h_y^- > \ts_t/2$ (i.e. a buyer with value $v_y^-$ is poor), we apply the following three steps on $\nzz$:
\begin{itemize}
\item[1.] Remove $\frac{w_y^-}{f_{\nd}(v_y^-)}$ fraction of the weight from all the singleton signals on $v_y^-$ and all the equal-revenue binary signals where $v_y^-$ is the taker. For this removed weight, collect their taker masses, and discard their giver masses. (Intuitively, we collect the rectangle $a_y^-$.)
\item[2.] Remove $\frac{w_y^+}{f_{\nd}(v_y^+)}\cdot \frac{h_y^+}{\ts_t+h_y^+}$ fraction of weight from all equal-revenue binary signals where $v_y^+$ is the taker. For the removed weight, collect their giver masses and discard their taker masses. (Intuitively, we collect the givers to the rectangle $a_y^+$.)
\item[3.] Build equal-revenue binary signals using the masses collected in Step 2 as giver masses and the masses collected in Step 1 as taker masses.
\end{itemize}

\begin{algorithm}[h]
\SetKwComment{Comment}{~$\triangleright$~}{}
 \SetKwInput{KwData}{Input}
\SetKwInput{KwResult}{Output}
\caption{Smoothing Algorithm\label{alg:smooth}}
\KwData{$\nzz=\{(\ns_q,\gamma_q)\}_{q\in [Q]}$, ironed surplus-mass function $\ts(x)$, ironing intervals $\{I_t\}_{t\in [T]}$.}
\KwResult{Smoothed signaling scheme $\nzm=\{(\ns'_{q'},\gamma'_{q'})\}_{q'\in [Q']}$.}
$\nzm\gets\nzz;$ $q'\gets Q$\Comment*[r]{\textrm{Initialize  $\nzm$ to $\nzz$.}}
Denote current $\nzm$ by $\{\ns'_{q},\gamma'_{q}\}_{q\in [Q]}$\;
\For {$t=1$ to $T$} {
Apply \cref{alg:rectangle} to find rectangle pairs $R_t=\{(a_y^+, a_y^-)\}_{y\in[Y_t]}$ on ironing interval $I_t$\;
Width and height of $a_y^+$ (resp. $a_y^-$) are $w_y^+$ and $h_y^+$ (resp. $w_y^-$ and $h_y^-$); $a_y^+$ (resp. $a_y^-$) are occupied by buyers with value $v_y^+$ (resp. $v_y^-$)\;
\For {each rectangle pair $(a_y^+, a_y^-)$ in $R_t$ with $h_y^->\ts_t/2$\label{alg_line:for}} {
    \For {each $\ns_q$ in $\nzz$ that is an equal-revenue binary signal with taker $v_y^-$, or is a singleton signal on $v_y^-$ \label{alg_line:middle_loop_start}} {
        $\gamma'_{q}\gets \gamma'_{q}-\gamma_q\cdot \frac{w_y^-}{f_{\nd}(v_y^-)}$\; \DontPrintSemicolon\Comment*[r]{\textrm{Step 1: Collect $\gamma_q\cdot \frac{w_y^-}{f_{\nd}(v_y^-)}\cdot f_{\ns_q}(v_y^-)$ taker mass from value $v_y^-$}} \label{alg_line:taker}
    }
    \For {each $\ns_q$ in $\nzz$ that is an equal-revenue binary signal with taker $v_y^+$ and some giver $v'$ \label{alg_line:for_2_start}} {
        $\gamma'_{q}\gets \gamma'_{q}-\gamma_q\cdot \frac{w_y^+}{f_{\nd}(v_y^+)}\cdot \frac{h_y^+}{\ts_t+h_y^+}$\; \DontPrintSemicolon\Comment*[r]{\textrm{Step 2: Collect $\gamma_q\cdot \frac{w_y^+}{f_{\nd}(v_y^+)}\cdot \frac{h_y^+}{\ts_t+h_y^+}\cdot f_{\ns_q}(v')$ giver mass from value $v'$}}\label{alg_line:ext_giver} \PrintSemicolon
        $q'\gets q'+1$\;
        $\gamma'_{q'}=\gamma_q\cdot \frac{w_y^+}{f_{\nd}(v_y^+)}\cdot \frac{h_y^+}{\ts_t+h_y^+}\cdot \big(1-\frac{v'}{v_y^+}\big)\big/\big(1-\frac{v'}{v_y^-}\big)$\;
        Add equal-revenue binary signal $\ns'_{q'}=\ns_{v',v_y^-}^\mathrm{E}$ with weight $\gamma'_{q'}$\ to $\nzm$\; \DontPrintSemicolon\Comment*[r]{\textrm{Step 3: Build equal-revenue binary signal on $(v',v_y^-)$ with collected mass}} \label{alg_line:new_bin}
    }\label{alg_line:for_2_end}
}
}
Make all the remaining masses singleton signals and add them to $\nzm$\;
\end{algorithm}

We formally present the smoothing process in \cref{alg:smooth}. Denote its output by $\nzm=\{(\ns'_{q'}, \gamma'_{q'})\}_{q'\in [Q']}$. 

\subsubsection{Analysis: Proof of \cref{thm:smooth}}
We show that $\nzm$ is feasible in \cref{lem:feasible_smooth}. In other words, we prove that the outcome of the algorithm satisfies Bayes plausibility as in \cref{eq:balance}. 

\begin{lemma} 
\label{lem:feasible_smooth}
\label{lem:exhaust}
    $\forall\,i\in [n],\ \sum_{q'\in [Q']} f_{\ns'_{q'}}(v_i)\cdot \gamma'_{q'}=f_{\nd}(v_i)$. 
\end{lemma}
\begin{proof} 
Since \cref{alg:smooth} makes all the remaining masses into singleton signals and adds them to $\nzm$ at the end, we simply need to show that the sum of masses in the equal-revenue binary signals does not exceed the  mass of the prior for each value. Since we only add new signals in Line~\ref{alg_line:new_bin}, it suffices to argue that there is always enough mass on $v'$ and $v_y^-$ to build equal-revenue binary signals. 

We first argue that the total mass at $v'$ is preserved. Consider a single run of Line~\ref{alg_line:ext_giver} to Line~\ref{alg_line:new_bin} with fixed $(t,y,q)$. In Line~\ref{alg_line:ext_giver}, the mass on $v'$ is reduced by
\[
f_{\ns_q}(v')\cdot \gamma_q\cdot \frac{w_y^+}{f_{\nd}(v_y^+)}\cdot \frac{h_y^+}{\ts_t+h_y^+}=\left(1-\frac{v'}{v_y^+}\right)\cdot \gamma_q\cdot \frac{w_y^+}{f_{\nd}(v_y^+)}\cdot \frac{h_y^+}{\ts_t+h_y^+}.
\]
This is exactly the mass added to $v'$ on the newly constructed signal (i.e. $\gamma'_{q'}\cdot \big(1-\frac{v'}{v_y^-}\big)$), thus the mass  on $v'$ is preserved. 

We next argue the mass collected from $v_y^-$ is enough for the new equal-revenue binary signals. Consider a single run of Line~\ref{alg_line:middle_loop_start} to Line~\ref{alg_line:for_2_end} with fixed $(t,y)$. In $\nzz$, at least $\frac{1}{2}\cdot f_{\nd}(v_y^-)$ mass from value $v_y^-$ is devoted to singletons signals on $v_y^-$ or  binary signals where $v_y^-$ is a taker. Therefore, the  mass collected from $v_y^-$ in Line~\ref{alg_line:taker} is at least 
\[
\frac{1}{2}\cdot f_{\nd}(v_y^-)\cdot \frac{w_y^-}{f_{\nd}(v_y^-)}=w_y^-/2.
\]
Let $r(t,y)=\frac{w_y^+}{f_{\nd}(v_y^+)}\cdot \frac{h_y^+}{\ts_t+h_y^+}$. Consider the new equal-revenue signals added in Line~\ref{alg_line:new_bin} with $v_y^-$ as taker. We assume that the $j^{\text{th}}$ equal-revenue binary signal added to $\nzz$ during \cref{alg:match} is $\ns_j=\ns_{\vg_j,\vt_j}^\mathrm{E}$ with weight $\gamma_j$. Let $\tg_j=f_{\ns_j}(\vg_j)\cdot \gamma_j$ and $\tt_j=f_{\ns_j}(\vt_j)\cdot \gamma_j$.  Denote the sum of their taker masses by $\nm_y^-$. We have:
\begin{align*}
\nm_y^-&=\sum_{j:\,\vt_j=v_y^+} \gamma_j\cdot r(t,y) \cdot \frac{\vg_j}{v_y^-} \cdot \left(1-\frac{\vg_j}{\vt_j}\right)\Big/\left(1-\frac{\vg_j}{v_y^-}\right)\\
&=\sum_{j:\,\vt_j=v_y^+} r(t,y)\cdot \tt_j\cdot \frac{\vt_j-\vg_j}{v_y^- - \vg_j}.  \tag{Since $\ns_j$ is an equal-revenue binary signal, $f_{\ns_j}(\vt_j) = \frac{\vg_j}{\vt_j}$}
\end{align*}
Since $\vt_j=v_y^+ < v_y^-$ and since at most $\frac{1}{2}\cdot f_{\nd}(v_y^+)$ mass from $v_y^+$ is used as taker mass in equal-revenue binary signals in $\nzz$, the above simplifies to:
\[
\nm_y^-\le \sum_{j:\,\vt_j=v_y^+} r(t,y)\cdot \tt_j
\le \frac{f_{\nd}(v_y^+)}{2}\cdot  \frac{w_y^+}{f_{\nd}(v_y^+)}\cdot \frac{h_y^+}{\ts_t+h_y^+} \le \frac{w_y^+}{2}\cdot \frac{h_y^+}{\ts_t}.\]
Since we have $w_y^+\cdot h_y^+=w_y^-\cdot h_y^-$ (paired rectangles have the same area) and since $h_y^->\ts_t/2$, it follows that $\nm_y^-$ is upper bounded by $w_y^-/2$. Therefore, the mass on $v_y^-$ collected in Line~\ref{alg_line:taker} is enough for constructing all the new binary signals. The outcome of \cref{alg:smooth} is therefore a feasible signaling scheme. 
\end{proof}
 
We now show the following lemma  that lower bounds each buyer's expected consumer surplus in $\nzm$ by half of the ironed expected consumer surplus: 
\begin{lemma} \label{lem:smooth}
For any $x\in(0,1]$, $s_{\nzm}(x)\ge \ts(x)/2$.
\end{lemma}
\begin{proof}
Consider a single iteration of the loop Line~\ref{alg_line:middle_loop_start} to Line~\ref{alg_line:for_2_end}, with fixed $(t,y)$. For each equal-revenue binary signal $\ns_j$ in $\nzz$ with $v_y^+$ as taker, we have collected $\frac{w_y^+}{f_{\nd}(v_y^+)}\cdot \frac{h_y^+}{\ts_t+h_y^+}$ fraction of the giver mass from it in Line~\ref{alg_line:ext_giver}. 
When we combine this giver mass (denoted by $M_j^{\mathrm{G}}$) with taker mass on value $v_y^-$ to form an equal-revenue binary signal $\ns'_{q'}=\ns_{\vg_j,v_y^-}^\mathrm{E}$, the taker mass in this binary signal (denoted by $M_j^{\mathrm{T}}$) is thus
\begin{align*}
M_j^{\mathrm{T}}=M_j^{\mathrm{G}}\cdot \frac{\vg_j}{v_y^--\vg_j}&=\frac{w_y^+}{f_{\nd}(v_y^+)}\cdot \frac{h_y^+}{\ts_t+h_y^+}\cdot \tt_j\cdot \frac{\vt_j-\vg_j}{\vg_j}\cdot \frac{\vg_j}{v_y^--\vg_j}\\
&=\frac{w_y^+}{f_{\nd}(v_y^+)}\cdot \frac{h_y^+}{\ts_t+h_y^+}\cdot \tt_j\cdot \frac{\vt_j-\vg_j}{v_y^--\vg_j}.
\end{align*}
Since buyers with value $v_y^-$ gain surplus $(v_y^--\vg_j)$ in this binary signal, the contribution from this signal to the total surplus of the buyers with value $v_y^-$ is 
\[
M_j^{\mathrm{T}} \cdot (v_y^--\vg_j)=\frac{w_y^+}{f_{\nd}(v_y^+)}\cdot \frac{h_y^+}{\ts_t+h_y^+}\cdot \tt_j\cdot (\vt_j-\vg_j).
\]

Again assume the $j^{\text{th}}$ equal-revenue binary signal added to $\nzz$ during \cref{alg:match} is $\ns_j=\ns_{\vg_j,\vt_j}^\mathrm{E}$ with weight $\gamma_j$. Let $\tg_j=f_{\ns_j}(\vg_j)\cdot \gamma_j$ and $\tt_j=f_{\ns_j}(\vt_j)\cdot \gamma_j$. Thus the total surplus of buyers with value $v_y^-$ in the newly constructed signals ($\ns'_{q'}$) is:
\begin{align*}
\sum_{j:\,\vt_j=v_y^+} \frac{w_y^+}{f_{\nd}(v_y^+)}\cdot \frac{h_y^+}{\ts_t+h_y^+}\cdot \tt_j\cdot (\vt_j-\vg_j)&= \frac{w_y^+}{f_{\nd}(v_y^+)}\cdot \frac{h_y^+}{\ts_t+h_y^+}\cdot\sum_{j:\,\vt_j=v_y^+} \tt_j\cdot (\vt_j-\vg_j)\\
&= \frac{w_y^+}{f_{\nd}(v_y^+)}\cdot \frac{h_y^+}{\ts_t+h_y^+} \cdot(\ts_t+h_y^+)\cdot f_{\nd}(v_y^+) \tag{By the definition of buyer surplus}\\
&= w_y^-\cdot h_y^-.\tag{Since paired rectangles have the same area}
\end{align*}

We now argue that, after the smoothing process, the expected consumer surplus of buyers with any value in the ironing interval $I_t$ is at least $\ts_t/2$. We consider three cases based on the expected surplus value of the  buyer in $\nzz$:

\begin{itemize}
\item[1.] Suppose that before the smoothing process, buyers with value $v_i$ have surplus strictly less than $\ts_t/2$. The corresponding area below $\ts_t$ (on the mass coordinates of these buyers) consists of rectangles $a_{y_1}^-$ to $a_{y_\psi}^-$. Since all these rectangles participate in the smoothing process, we can sum up the surplus after smoothing from $y=y_1$ to $y=y_\psi$ as 
\begin{equation*}
\sum_{y=y_1}^{y_\psi} w_y^-\cdot h_y^-=m_{i}\cdot h_y^{-}>m_i\cdot \ts_t/2.\tag{Since $\sum_{y=y_1}^{y_\psi} w_y^-=m_i$ and $h_y^->\ts_t/2$}
\end{equation*}

Dividing by total mass $m_i$, the expected surplus of these buyers is at least $\ts_t/2$.

\item[2.] On the buyers with value $v_i$ who originally have surplus less than $\ts_t$ but at least $\ts_t/2$, we have not changed any equal-revenue signal which contains $v_i$ as taker value in the process. Their expected surplus does not change and is at least $\ts_t/2$. 
\item[3.] On the buyers with value $v_i$ who originally gain surplus more than $\ts_t$. Suppose their corresponding area above $\ts_t$  is split into rectangles $a_{y_1}^+$ to $a_{y_\xi}^+$. We have extracted at most $\sum_{y=y_1}^{y_\xi}\frac{w_y^+}{f_{\nd}(v_y^+)}\cdot \frac{h_y^+}{\ts_t+h_y^+}=\frac{h_{y_1}^+}{\ts_t+h_{y_1}^+}$ fraction of the weight of their original signals. After the smoothing process, their expected surplus is still at least  $\frac{\ts_t}{\ts_t+h_{y_1}^+}\cdot(\ts_t+h_{y_1}^+)=\ts_t$. \\
\end{itemize}
Combining the above three cases, on each ironing interval $I_t$, in all buyers have surplus at least $\ts_t/2$ in $\nzm$. This completes the proof. 
\end{proof}

Finally, we make the scheme monotone. Suppose that $s_{\nzm}(x)>\ts(x)/2$ for some $x$ corresponding to buyers of value $v$. For each equal-revenue binary signal in $\nzm$ where $v$ is the taker value, we remove $\frac{2s_{\nzm}(x)-\ts(x)}{2s_{\nzm}(x)}$ fraction of the weight of the signal into two singleton signals. Denote the signaling scheme after the decomposition by $\nzf$. Since the singleton signals do not provide buyer surplus, we have $$\cs_v(\nzf)=s_{\nzm}(x)\cdot \left(1-\frac{2s_{\nzm}(x)-\ts(x)}{2s_{\nzm}(x)}\right)=\ts(x)/2.$$ The surplus-mass function of buyers in $\nzf$ is exactly $\ts(x)/2$, completing the proof of \cref{thm:smooth}.

\subsubsection{Proof of \cref{thm:main}} \label{sec:proof_of_main}
Based on the construction of $\nzz, \nzm$, and $\nzf$, we prove our main theorem: $\nzf$ is efficient, monotone, and $8$-majorized.
\begin{proof}[Proof of \cref{thm:main}]
In the whole construction process of $\nzz$ and $\nzf$, we only include equal-revenue binary signals (signals where the induced posterior only has two supports and they yield the same revenue for the seller, see \cref{def:bin_sig} as a formal definition) or singleton signals (signals where the induced posterior only has one support). The seller will always select the lowest support in the posterior as the price in these two types of signals. Therefore, the item is always sold, and $\nz_1$ is efficient. 
Further, the surplus-mass function of $\nzf$ is $\ts/2$, which is monotonically increasing. 


Finally, we show that $\nz$ is 8-majorized. Consider an arbitrary scheme $\nz'$ and any value $m\in (0,1]$. Note that $\ts$ is a step function and it preserves the prefix sum of $s_{\nzz}$ at break-point coordinates $m_0,m_1,\ldots,m_\tau$ (i.e. $\forall\, i\in [\tau],\ \pfv(\ts, m_i)=\pfv(s_{\nzz}, m_i)$) in ascending order. We have $m_0=0$ and $m_\tau=1$. Since $0<m\le 1$, there exists a unique $k\in [\tau+1]$ such that 
 $$m_{k-1}< m\le m_{k}.$$ By setting $\lambda=\frac{m-m_{k-1}}{m_{k}-m_{k-1}}$, we have $$m=(1-\lambda)\cdot m_{k-1}+\lambda \cdot m_{k}.$$
Consider the sorted $m$-prefix sum of $s_{\nz'}$, we have 

\begin{align*}
\pf(s_{\nz'}, m)&\le \pf(s_{\nz'},m_{k-1})+\lambda\cdot\left( \pf(s_{\nz'},m_{k})-\pf(s_{\nz'},m_{k-1})\right) \tag{By convexity of $\pf$ as a function of $m$}\\
&= (1-\lambda)\cdot \pf(s_{\nz'},m_{k-1})+ \lambda \cdot \pf(s_{\nz'},m_{k})\\
&\le (1-\lambda)\cdot \pfv(s_{\nz'},m_{k-1})+ \lambda \cdot \pfv(s_{\nz'},m_{k}) \tag{Prefix sum is at least sorted prefix sum}\\
&\le 4 \cdot \left[(1-\lambda)\cdot\pfv(s_{\nzz},m_{k-1}) + \lambda \cdot \pfv(s_{\nzz},m_{k})\right]
\tag{By \cref{thm:value_pr}}\\
&= 4\cdot \left[(1-\lambda)\cdot\pfv(\ts,m_{k-1}) + \lambda \cdot \pfv(\ts,m_{k})\right] \tag{Since $\ts$ preserves the prefix sum of $s_{\nzz}$ at $m_{k-1}$ and $m_k$}\\
&= 4\cdot \left[(1-\lambda)\cdot\pf(\ts,m_{k-1}) + \lambda \cdot \pf(\ts,m_{k})\right] \tag{Since $\ts$ is monotone}\\
&\le 4\cdot \pf(\ts,m) \tag{Since $\ts$ is constant on $(m_{k-1},m_{k}]$}\\
&= 8\cdot \pf(s_{\nzf},m). \tag{By \cref{thm:smooth}}
\end{align*}
Therefore, by \cref{def:approx_maj}, $\nzf$ is 8-majorized.
\end{proof}

\section{Lower Bounds}
\label{sec:lb}

Finally, we complement our $8$-majorized signaling scheme with two lower bounds for finding $\alpha$-majorized signaling schemes.
Our first bound shows the impossibility of $\alpha$-majorization for any constant $\alpha$ if we restrict to \emph{buyer-optimal} signaling schemes, and the second shows the impossibility of $\alpha$-majorization for $\alpha < 1.5$. Both our hard instances are, in a sense, the simplest possible -- they involve distributions $\nd$ over only three different values. (For distributions over two values, there is an exactly majorized scheme, since agents with the lower value always get surplus $0$.)

\subsection{Incompatibility of Approximate Majorization and Buyer Optimality}
We now show that no buyer-optimal scheme (i.e., one that maximizes $\cs(\nz)$) can yield our guarantees; in particular, no buyer-optimal scheme can be $\alpha$-majorized, for any constant $\alpha \in \mathbb{R}^+$. This also shows that exact majorization is impossible since exact majorization implies buyer optimality. This motivates the need for approximations and the need for looking beyond buyer-optimal schemes.

Recall that a signaling scheme $\nz$ is \emph{buyer-optimal} if $
\sum_{i=1}^n f_{\nd}(v_i)\cdot \cs_{v_i}(\nz)=
\sum_{i=1}^n f_{\nd}(v_i)\cdot v_i-\nr^{\textsf{My}}(\nd),
$ where $\nr^{\textsf{My}}(\nd)=\max_{v} v\cdot G_{\nd}(v)$ is the optimal (Myerson) revenue under $\nd$ (i.e. without signaling).
The following observation is immediate.
\begin{lemma}\label{lem:price_remain}
If a scheme is buyer-optimal, any optimal price for the original distribution will remain optimal in any signal of the scheme.
\end{lemma}
\begin{proof}
Since the scheme is buyer-optimal, the item is always sold. The buyers' total surplus is thus the expected value of the buyer minus the seller's revenue. Suppose that $p^{\textsf{My}}$ is an optimal price in the original distribution $\nd$. The seller gets revenue $R^{\textsf{My}}=p^{\textsf{My}}\cdot G_{\nd}(p^{\textsf{My}})$ without signaling. 

Note that the seller can still gain $R^{\textsf{My}}$ in the scheme if they nevertheless post $p^{\textsf{My}}$ for all the signals. Suppose for contradiction that in some signal, $p^{\textsf{My}}$ is not an optimal price, the seller must gain strictly more revenue from posting price $p'$ than that from posting $p^{\textsf{My}}$. Since in all other signals, the seller gains revenue at least as much as by posting price $p^{\textsf{My}}$, the seller's overall revenue is strictly greater than $R^{\textsf{My}}$. Therefore, $\nzz$ cannot be buyer-optimal, leading to a contradiction.
\end{proof}

To prove the lower bounds for any signaling scheme, we need the following lemma to narrow the space of signaling schemes we are considering. The lemma says that, for any signaling scheme, we can always transform it into an equivalent scheme (by which we mean that the expected consumer surplus of any buyer remains unchanged) with a simple form. \citet*{Cummings20} have a similar observation, but we need the following lemma that provides finer structural characterizations.

\begin{lemma}\label{lem:wlog}
For any signaling scheme $\nz$, there exists an \emph{efficient} signaling scheme $\nz^\mathrm{E}$, such that: 
\begin{itemize}
\item[1.] $\nz^\mathrm{E}$ generates the same expected consumer surplus as $\nz$ for any buyer; 
\item[2.] $\nz^\mathrm{E}$ includes at most $n$ signals, all of which have different lowest-supports (i.e. smallest value with non-zero mass).
\end{itemize}
\end{lemma}
\begin{proof}
We conduct two operations on $\nz$ to construct $\nz'$. First, for each signal in $\nz$, we discard all masses on values strictly less than the optimal price. Then we put all the discarded masses into at most $n$ singleton signals. Each singleton signal includes all the discarded masses on a different value $v_i$ ($i\in[n]$). Consider any signal in $\nz$ before this operation. Since any buyer with the discarded values does not gain surplus from this signal and after discarding the optimal price remains the same, the expected consumer surplus of any buyer does not change.

Second, we combine all groups of signals with the same smallest value in support into one signal by adding up the masses on each value. Let $\nz'$ be the signaling scheme with all the signals after the combinations. After each combination operation, the optimal price remains the same at the smallest support. Therefore, this combination process again does not change the expected consumer surplus of any buyer. Since there are in total $n$ different values, there are at most $n$ signals in $\nz'$, each corresponding to a different smallest support. 

Since the optimal price for each signal remains the lowest support after combination, the item is always sold. Therefore, $\nz'$ is efficient. 
\end{proof}

\begin{theorem}
For any given $\alpha \ge 1$, there exist instances under which no buyer-optimal scheme is $\alpha$-majorized. 
\end{theorem}
\begin{proof}
Consider the instance where there are three buyer values: $v_1=1$, $v_2=N$, $v_3=N+1$, and the probability masses of these values are
\[
f_\nd(v_1)=\frac{N^2-1}{N^3+2N^2+N},\ f_{\nd}(v_2)=\frac{N^2+1}{N^3+2N^2+N},\ f_{\nd}(v_3)=\frac{N^3+N}{N^3+2N^2+N}.
\]
We will show that no buyer-optimal scheme can be $\alpha$-majorized on this instance for $\alpha < N$. Suppose for the purpose of contradiction that $\nz$ is $\alpha$-majorized where $\alpha < N$.

Notice that $v_2$ and $v_3$ are both optimal prices for the seller. By \cref{lem:wlog}, we can transform $\nz$ into a signaling scheme $\nz'$ with at most three signals, each with a different smallest support. Since the transformation preserves any buyer's expected consumer surplus, $\nz'$ is still buyer-optimal.

By \cref{lem:price_remain}, in $\nz'$, we must include both $v_2$ and $v_3$ as the optimal prices. Therefore, there is no signal with $v_3$ as the smallest support. Since the signaling scheme is efficient (by buyer optimality), the smallest support in any signal must also be the optimal price. Therefore, we conclude that in $\nz'$, there are only two signals $\ns_1$ and $\ns_2$:
\begin{itemize}
\item[1.] The first signal $\ns_1$ includes all three values. Each value is an optimal price on the signal. Since each value as price provides the same revenue to the seller, the probability masses of $\ns_1$ on the three values must be 
\begin{equation*}
    f_{\ns_1}(v_1)=\frac{N^2-1}{N^2+N};\ f_{\ns_1}(v_2)=\frac{1}{N^2+N};\ f_{\ns_1}(v_3)=\frac{N}{N^2+N}.
\end{equation*}
Since only $\ns_1$ include $v_1$, we have $\gamma_1=f_{\nd}(v_1)\big/\big(\frac{N^2-1}{N^2+N}\big)=\frac{1}{N+1}$.  
\item[2.] The second signal includes only two values $v_2$ and $v_3$. Both of them are optimal prices. It is also an equal-revenue signal. The probability masses of $\ns_2$ on the three values are \begin{equation*}
    f_{\ns_1}(v_1)=0;\ f_{\ns_1}(v_2)=\frac{1}{N+1};\ f_{\ns_1}(v_3)=\frac{N}{N+1}.
\end{equation*}
Since $\gamma_2+\gamma_1=1$, we have $\gamma_2=\frac{N}{N+1}$.
\end{itemize}

The expected consumer surplus from $\nz'$ of any buyer with value $v_2=N$ is \[\cs_{v_2}(\nz')=\frac{\gamma_1\cdot f_{\ns_1}(v_2)\cdot (N-1)}{f_{\nd}(v_2)}=\frac{N-1}{N^2+1}.\]

The expected consumer surplus from $\nz'$ of any buyer with value $v_3=N+1$ is \[
\cs_{v_3}(\nz')=\frac{\gamma_1\cdot f_{\ns_1}(v_3)\cdot N+\gamma_2\cdot f_{\ns_2}(v_3)\cdot 1}{f_{\nd}(v_3)}=\frac{N+N^2}{N^2+1}. 
\]

We consider another scheme $\nz_2$ consisting of three signals $\ns_1'$, $\ns_2'$ and $\ns_3'$:
\begin{gather*}
f_{\ns_1}(v_1)=\frac{N^2-1}{N^2+N};\ f_{\ns_1}(v_2)=\frac{N+1}{N^2+N};\ f_{\ns_1}(v_3)=0;\ \gamma_1=\frac{1}{N+1};\\
f_{\ns_2}(v_1)=0;\ f_{\ns_2}(v_2)=\frac{N^2-N}{N^3-N};\ f_{\ns_2}(v_3)=\frac{N^3-N^2}{N^3-N};\ \gamma_2=\frac{N-1}{N+1};\\
f_{\ns_3}(v_1)=0;\ f_{\ns_3}(v_2)=0;\ f_{\ns_3}(v_3)=1;\ \gamma_3=\frac{1}{N+1}.
\end{gather*}

In $\nz_2$, the expected consumer surplus of any buyer with value $v_2=N$ is
\[
\cs_{v_2}(\nz_2)=\frac{(N+1)\cdot (N-1)}{N^2+1}=\frac{N^2-1}{N^2+1}; 
\]

The expected consumer surplus of any buyer with value $v_3=N+1$ is
\[
\cs_{v_3}(\nz_2)=\frac{(N^3-N^2)\cdot 1}{N^3+N}=\frac{N^2-N}{N^2+1}; 
\]

Therefore, the smallest non-zero expected consumer surplus of a buyer in any buyer-optimal scheme is $\frac{N-1}{N^2+1}$ and the smallest non-zero expected consumer surplus of a buyer in $\nz_2$ is $\frac{N^2-N}{N^2+1}$. Since we have $\big(\frac{N^2-N}{N^2+1}\big)/\big(\frac{N-1}{N^2+1}\big)=N>\alpha$, $\nz$ is not $\alpha$-majorized by $\nz_2$, leading to a contradiction.
\end{proof}

\subsection{Lower Bound for Approximate Majorization}
Finally we provide a lower bound on approximate majorization under general signaling schemes. The following theorem shows that no scheme can be better than $1.5$-majorized. This complements our upper bound of $8$-majorization.

\begin{theorem}
For any $\alpha < 1.5$, there exist instances where no signaling scheme is $\alpha$-majorized.
\end{theorem}
\begin{proof}
Suppose that there exists a signaling scheme $\nz$ that is $\alpha$-majorized. Let $\alpha=1.5-\delta$, where $\delta>0$. Set $\eps$ such that $0<\eps\ll \delta$. There are three buyer values $\langle v_1, v_2, v_3\rangle=\langle 1, 1 + \eps, 2 + \eps\rangle$. The probability masses of $\nd$ on these values are $f_{\nd}(v_1)=\eps^2+2\eps$, $f_{\nd}(v_2)=1+(1+\eps)^2$, $f_{\nd}(v_3)=(1+\eps)+(1+\eps)^3$. For simplicity, we omit the multiplicative factor of normalizing the whole population to 1. 

First, we compute the largest expected consumer surplus of the smallest-surplus buyer. By \cref{lem:wlog}, we can transform $\nz$ into $\nz'$ without changing any buyer's expected consumer surplus. Assume $\nz'$ consists of the following three signals, each represented by the mass vector $\gamma_q\cdot \ns_q$: $$\gamma_1 \cdot \ns_1 = \langle x, y, z \rangle;\ \gamma_2 \cdot \ns_2 = \langle 0, y', z'\rangle;\ \gamma_3 \cdot \ns_3 = \langle 0, 0, z''\rangle.$$ 
The expected consumer surplus of a buyer with value $v_2$ is $\frac{y\cdot \eps}{f_{\nd}(v_2)}$. The expected consumer surplus of a buyer with value $v_3$ is $\frac{z\cdot (1+\eps)+z'\cdot 1}{f_{\nd}(v_3)}$. The maximum of the smallest expected consumer surplus of a buyer with value $v_2$ or $v_3$ can be solved by the following LP:
\begin{gather}
    \text{Maximize $s_{\min}$},\ \text{s.t.}\notag\\
    \frac{y\cdot \eps}{f_{\nd}(v_2)}\ge s_{\min};\label{con:1}\\ \frac{z\cdot (1+\eps)+z'\cdot 1}{f_{\nd}(v_3)}\ge s_{\min};\label{con:2}\\
    1\cdot (x+y+z)\ge (1+\eps)\cdot (y+z);\label{con:3}\\
    1\cdot (x+y+z)\ge (2+\eps)\cdot z;\label{con:4}\\
    (1+\eps)\cdot (y'+z')\ge (2+\eps)\cdot z';\label{con:5}\\
    0\le x\le f_{\nd}(v_1); \label{con:6}\\
    y+y'\le f_{\nd}(v_2);\label{con:7}\\ 
    z+z'+z''\le f_{\nd}(v_3);\label{con:8}\\ y,z,y',z',z''\ge 0.\notag
\end{gather}

\cref{con:1} and \cref{con:2} mean that the minimum expected consumer surplus of any buyer with value $v_2$ or $v_3$ is at least the objective function. \cref{con:3} and \cref{con:4} are constraints on the masses of the first signal so that
$1$ is the optimal price. Similarly, \cref{con:5} ensures that $v_2=1+\epsilon$ is the optimal price in the second signal. \cref{con:6,con:7,con:8} ensures that for any $i\in[3]$, the sum of the masses on a value $v_i$ in all three signals does not exceed the total mass $m_i$. 

One feasible solution of the LP has value 
\[s_{\min}^*=\frac{4+3\eps+\eps^2}{2+\eps}\cdot \frac{\eps}{f_{\nd}(v_2)}>\frac{2\eps}{f_{\nd}(v_2)},\]
and is obtained when 
\begin{gather*}
x=f_{\nd}(v_1)=\eps^2+2\eps;\\
y=\frac{4+3\eps+\eps^2}{2+\eps};\\
z=\frac{\eps}{2+\eps};\\
y'=f_{\nd}(v_2)-y=\frac{\eps^3+3\eps^2+3\eps}{2+\eps};\\
z'=y'\cdot (1+\eps)=\frac{(1+\eps)\cdot(\eps^3+3\eps^2+3\eps)}{2+\eps}.
\end{gather*}

If $\nz'$ is $\alpha$-majorized (and recall that $\alpha=1.5-\delta$), the smallest positive expected consumer surplus of any buyer is at least
\[
\frac{s_{\min}^*}{\alpha}=s_{\min}^*\cdot \frac{1}{1.5-\delta}> \frac{2\eps}{f_{\nd}(v_2)}\cdot \frac1{1.5-\delta}\ge \frac{2\eps}{f_{\nd}(v_2)}\cdot\left(\frac23+\frac49 \delta\right).
\]
Therefore, we have that the expected consumer surplus of any buyer with value $v_2$ from $\nz'$ is at least $\frac{s_{\min}^*}{\alpha}$:
\[
\cs_{v_2}(\nz')= \frac{y\cdot \eps}{f_{\nd}(v_2)}\ge \frac{s_{\min}^*}{\alpha}\ge \frac{2\eps}{f_{\nd}(v_2)}\cdot\left(\frac23+\frac49 \delta\right),
\]
and thus,
\begin{equation}
y\ge \frac43+\frac89\delta.\label{con:y}
\end{equation} 
By \cref{con:7}: $y+y'\le f_{\nd}(v_2)=2+2\eps+\eps^2$, and we have 
\begin{equation}
y'\le 2+2\eps+\eps^2-\frac43-\frac89\delta=\frac23+2\eps+\eps^2-\frac89\delta.
\label{eqn:s_surplus}
\end{equation}
By \cref{con:3}: $x+y+z\ge (1+\eps)\cdot (y+z)$, we have $x\ge \eps\cdot(y+z)$. Since $x\le f_{\nd}(v_1)=\eps^2+2\eps$, we have $y+z\le 2+\eps$ and thus $$z\le 2+\eps-y\le 2+\eps-\left(\frac43+\frac89\delta\right)=\frac23+\eps-\frac89\delta.$$ 
By \cref{eqn:s_surplus} and \cref{con:5}, we have 
\[
z'\le (1+\eps)\cdot y'\le (1+\eps)\cdot \left(\frac23+2\eps+\eps^2-\frac89\delta \right).
\]
Consider the overall surplus of all buyers from $\nz'$ (i.e. $\pf(\nz', \sum_{i\in[3]} f_{\nd}(v_i)$). We have:
\begin{align*}
\pf\Big(\nz', \sum_{i\in[3]} f_{\nd}(v_i)\Big)&=\cs_{v_2}(\nz') \cdot f_{\nd}(v_2)+\cs_{v_3}(\nz') \cdot f_{\nd}(v_3)\\&=y\cdot \eps+z\cdot (1+\eps)+z'\cdot 1\\
&\le (2+2\eps+\eps^2)\cdot \eps+\left(\frac23+\eps-\frac89\delta\right)\cdot (1+\eps)+(1+\eps)\cdot \left(\frac23+2\eps+\eps^2-\frac89\delta \right)\\
&= \frac{4}3-\frac{16}{9}\delta+o(\delta)\tag{Since $\eps\ll \delta$}\\
&<\frac43.
\end{align*}
Since $v_3$ is the optimal price on the original distribution, the optimal overall consumer surplus is:
\begin{align*}
f_{\nd}(v_1)\cdot v_1+f_{\nd}(v_2)\cdot v_2+f_{\nd}(v_3)\cdot v_3-f_{\nd}(v_3)\cdot v_3
&=f_{\nd}(v_1)+(1+\eps)\cdot f_{\nd}(v_2)\\
&=(\eps^2+2\eps)+(1+\eps)(2+2\eps+\eps^2)\\
&>2. 
\end{align*}
Since $\left(\frac4{3}\right)\big/ 2<\frac{1}{1.5-\delta}$, $\nz'$ cannot achieve $\alpha$-approximation on the optimal overall consumer surplus. Therefore, $\nz$ is not $\alpha$-majorized, leading to a contradiction.
\end{proof}

\bibliographystyle{abbrvnat}
\bibliography{ref}

\appendix
\section{Omitted Proofs}
\label{app:omitted}
\begin{proof}[Proof of \cref{prop:majorize_concave}]
Let $\vec{u}$ be the vector of expected utilities of $\nz$. Let $\vec{w}$ be that vector of any signaling scheme. We wish to show $\W(\vec{u}) \geq \frac{1}{\alpha} \cdot \W(\vec{w})$.

By normalization (or non-negativity) and concavity, $\W(\vec{u}) \geq \frac{1}{\alpha} \cdot \W(\alpha \vec{u})$. There exists a $\vec{u}'$ (by starting from $\alpha \vec{u}$ and gradually decreasing its largest elements) such that $\vec{u}' \leq \alpha \vec{u}$, $\Vert \vec{u}' \Vert_1 = \Vert \vec{w} \Vert_1$, and $\vec{u}'$ is majorized by $\vec{w}$. Therefore, $\W(\vec{u}) \geq \frac{1}{\alpha} \cdot \W(\alpha \vec{u}) \geq \frac{1}{\alpha} \cdot \W(\vec{u}') \geq \frac{1}{\alpha} \cdot \W(\vec{w})$, where the penultimate inequality is because the welfare function is weakly increasing, and the last inequality is from Schur-concavity (see e.g. \citep{convex_functions}) of $\W$, implied by symmetry and concavity.

For the converse, notice that for any $m \in [0, 1]$, the sorted $m$-prefix sum is welfare function that is symmetric, non-decreasing, concave and normalized (or non-negative).
\end{proof}

\end{document}